\documentclass[twocolumn]{autart}    
\usepackage[numbers,compress]{natbib}
\usepackage{amsmath, amssymb,mathrsfs, dsfont} 
\usepackage{amsmath,amssymb,bm,bbm}
\usepackage{ color, bm,bbm} 
\usepackage{ graphicx,graphics, epstopdf} 
\usepackage{algorithm, algorithmicx} 
\usepackage{tikzsymbols}
\usepackage{caption}
\DeclareMathAlphabet{\mathbbold}{U}{bbold}{m}{n}

\usepackage[english]{babel}
\usepackage[noend]{algpseudocode} 
\usepackage{tabularx}
\usepackage{booktabs}
\usepackage{multirow}
\usepackage{siunitx}
\usepackage{courier}
\usepackage{float}
\usepackage{makecell}

\newcommand{\ra}[1]{\renewcommand{\arraystretch}{#1}}

\usepackage{caption}
\usepackage{xspace}

\newcommand{\mtiny}[1]{{\scalebox{.75}{#1}}}
\newcommand{\smtiny}[1]{{\scalebox{.63}{#1}}}

\newcommand*{\supOp}{\operatornamewithlimits{sup}\limits}

\newcommand*{\esssup}{\operatornamewithlimits{ess\,sup}}

\newcommand*{\CupOp}{\operatornamewithlimits{\text{\scalebox{1.25}{$\cup$}}}\limits}

\newcommand{\Sum}{\sum\limits}
\newcommand{\GP}{\mathcal{G\!P}}

\newcommand{\tr}{{\smtiny{$\mathsf{T}$ }}\!}

\newcommand{\zero}{\mathbf{0}}

\newcommand{\onefun}{\mathbbm{1}}
\iftrue

\else

\fi

\newcommand{\vc}[1]{{ \mathrm{#1} }}
\newcommand{\mx}[1]{{ \mathrm{#1} }}

\newcommand{\drm}{\mathrm{d}}

\newcommand{\inner}[2]{{ \langle {#1,#2} \rangle}}
\newcommand{\nth}{{\text{\tiny{th}}}}

\renewcommand{\emptyset}{\varnothing}
\newcommand{\expe}{\mathrm{e}}

\newcommand{\floor}[1]{\lfloor #1\rfloor}

\newcommand{\ellone}{\ell^{1}}
\newcommand{\ellinfty}{\ell^{\infty}}

\newcommand{\Lone}{L^{1}}
\newcommand{\Linfty}{L^{\infty}}

\newcommand{\Lscrone}{\Lscr^{1}}
\newcommand{\Lscrinfty}{\Lscr^{\infty}}
\newcommand{\Lscrp}{\Lscr^{p}}

\newcommand{\Bscr}{{\mathscr{B}}}

\newcommand{\Fscr}{{\mathscr{F}}}
\newcommand{\Gscr}{{\mathscr{G}}}
\newcommand{\Hscr}{{\mathscr{H}}}

\newcommand{\Lscr}{{\mathscr{L}}}
\newcommand{\Mscr}{{\mathscr{M}}}

\newcommand{\Sscr}{{\mathscr{S}}}

\newcommand{\Fcal}{{\mathcal{F}}}

\newcommand{\Ical}{{\mathcal{I}}}

\newcommand{\Lcal}{{\mathcal{L}}}

\newcommand{\Ncal}{{\mathcal{N}}}

\newcommand{\Rcal}{{\mathcal{R}}}

\newcommand{\Xcal}{{\mathcal{X}}}
\newcommand{\Ycal}{{\mathcal{Y}}}

\newcommand{\Cbb}{{\mathbb{C}}}

\newcommand{\Ebb}{{\mathbb{E}}}

\newcommand{\Nbb}{{\mathbb{N}}}
\newcommand{\Pbb}{{\mathbb{P}}}

\newcommand{\Rbb}{{\mathbb{R}}}

\newcommand{\Tbb}{{\mathbb{T}}}

\newcommand{\Xbb}{{\mathbb{X}}}
\newcommand{\Ybb}{{\mathbb{Y}}}
\newcommand{\Zbb}{{\mathbb{Z}}}

\newcommand{\bfalpha}{\boldsymbol{\alpha}}

\newcommand{\bflambda}{\boldsymbol{\lambda}}

\input{mystyleproof}
\theoremstyle{plain}
\newtheorem{theorem}{Theorem}
\newtheorem{definition}{Definition}

\newtheorem{corollary}[theorem]{Corollary}
\newtheorem{lemma}[theorem]{Lemma}

\iftrue 
\usepackage{xcolor}
\usepackage[colorlinks=true]{hyperref}
	\hypersetup{
		linkcolor=blue,
		citecolor={red!55!black},
		urlcolor={green!80!black}
	}
\edef\endfrontmatter{%
	\unexpanded\expandafter{\endfrontmatter}
	\noexpand\endNoHyper 
}

	\allowdisplaybreaks
	
\fi

\newcommand{\vcb}{\vc{b}}

\newcommand{\vcg}{\vc{g}}

\newcommand{\vcm}{\vc{m}}

\newcommand{\vcs}{\vc{s}}

\newcommand{\vcu}{\vc{u}}
\newcommand{\vcv}{\vc{v}}

\newcommand{\mxL}{\mx{L}}

\newcommand{\Hilbert}{\Hscr}
\newcommand{\Banach}{\Bscr}

\newcommand{\kernel}{\mathds{k}}
\newcommand{\kernelh}{\mathds{h}}

\newcommand{\Hk}{\Hilbert_\mathbbm{k}}

\newcommand{\TC}{\text{\mtiny{$\mathrm{TC}$}}}
\newcommand{\DI}{\text{\mtiny{$\mathrm{DI}$}}}
\newcommand{\DC}{\text{\mtiny{$\mathrm{DC}$}}}
\renewcommand{\SS}{\text{\mtiny{$\mathrm{SS}$}}}
\newcommand{\iTC}{\text{\mtiny{$\mathrm{iTC}$}}}
\newcommand{\iSS}{\text{\mtiny{$\mathrm{iSS}$}}}
\newcommand{\iTS}{\text{\mtiny{$\mathrm{iTS}$}}}
\newcommand{\AMLS}{\text{\mtiny{$\mathrm{AMLS}$}}}
\newcommand{\st}{\text{\mtiny{$\mathrm{st}$}}}

\newcommand{\kernelTC}{\kernel_{\TC}}
\newcommand{\kernelDI}{\kernel_{\DI}}
\newcommand{\kernelDC}{\kernel_{\DC}}
\newcommand{\kernelSS}{\kernel_{\SS}}
\newcommand{\kerneliTC}{\kernel_{\iTC}}
\newcommand{\kerneliSS}{\kernel_{\iSS}}
\newcommand{\kerneliTS}{\kernel_{\iTS}}
\newcommand{\kernelAMLS}{\kernel_{\AMLS}}
\newcommand{\kernelst}{\kernel_{\st}}
\newcommand{\kernelSI}{\kernel_{\text{\mtiny{$\mathrm{SI}$}}}}
\newcommand{\kernelSSn}[1]{\kernel_{\SS{#1}}}
\newcommand{\kernelRnE}[1]
{\kernel_{\text{\mtiny{$\mathrm{R}$}}#1\text{\mtiny{$\mathrm{E}$}}}}

\newcommand{\DSRI}{\text{\mtiny{$\mathrm{DSRI}$}}}
\newcommand{\DSRIclass}{\Sscr_{\DSRI}}
\newcommand{\DSRIvalue}{\Mscr}
\newcommand{\Sstable}{\Sscr_{\mathrm{s}}}
\newcommand{\Sfinitetrace}{\Sscr_{\mathrm{ft}}}

\newcommand{\Lu}[1]{\mx{L}^{\!\vc{u}}_{#1}}

\renewcommand{\epsilon}{\varepsilon}

\newcommand{\Jimage}{\mathrm{j}}

\newcommand{\Fw}{\Fcal_\omega}
\newcommand{\Fwr}{\mx{F}_\omega^{\text{\rm{(r)}}}}
\newcommand{\Fwi}{\mx{F}_\omega^{\text{\rm{(i)}}}}

\newcommand{\GOmega}{\Gscr_{\Omega}}
\newcommand{\GT}{\Gscr_{\Tbb}}

\begin{document}
\begin{frontmatter}
\title{Diagonally Square Root Integrable Kernels in System Identification\vspace{-9mm}}
\thanks[footnoteinfo]{This paper was not presented at any IFAC 	meeting. Corresponding author M.~Khosravi.}

\author[First]{Mohammad Khosravi}\ead{mohammad.khosravi@tudelft.nl},  
\author[Second]{Roy S. Smith}\ead{rsmith@control.ee.ethz.ch}
\address[First]{Delft Center for Systems and  Control, Delft University of Technology}  
\address[Second]{Automatic Control Laboratory, ETH Z\"urich}  

\begin{keyword}                          
system identification; kernel-based methods; diagonally square root integrable kernels; stable Gaussian processes
\end{keyword}                            
\begin{abstract}              
In recent years, the reproducing kernel Hilbert space (RKHS) theory has played a crucial role in linear system identification. The core of a RKHS is the associated kernel characterizing its properties. Accordingly, this work studies the class of diagonally square root integrable (DSRI) kernels. We demonstrate that various well-known stable kernels introduced in system identification belong to this category. Moreover, it is shown that any DSRI kernel is also stable and integrable. We look into certain topological features of the RKHSs associated with DSRI kernels, particularly the continuity of linear operators defined on the respective RKHSs. For the stability of a Gaussian process centered at a stable impulse response, we show that the necessary and sufficient condition is the diagonally square root integrability of the corresponding kernel. Furthermore, we elaborate on this result by providing proper interpretations.
\vspace{-1mm}
\end{abstract}
\end{frontmatter}

\section{Introduction}\label{sec:introduction}
The theory of reproducing kernel Hilbert spaces (RKHSs) was introduced   \cite{aronszajn1950theory} midway through the 
twentieth century. The intrinsic properties of RKHSs, their one-to-one relationship with the positive definite kernels, and their fundamental ties to the Gaussian processes offer a strong foundation for addressing various estimation and interpolation problems 
\cite{parzen1959statistical,wahba1990spline,cucker2002best,berlinet2011reproducing,khosravi2021Koopman}. 
Accordingly, they have become increasingly prevalent in statistics, signal processing, learning theory, and numerical analysis  \cite{kimeldorf1970correspondence,lukic2001stochastic,kanagawa2018gaussian,cuckerANDsmale2002mathematical}.
On the other hand, system identification has emerged as the theory and techniques for estimating suitable mathematical representations of dynamical systems using measurement data \cite{zadeh1956identification}, and remained an active field of research by developing numerous methodologies  \cite{ljung2010perspectives, schoukens2019nonlinear,khosravi2021ROA,ahmadi2020learning,khosravi2021grad}.

The RKHS theory is brought to the system identification area in  \cite{pillonetto2010new} by developing \emph{kernel-based system identification methods}.
As a result, a paradigm shift occurred in the system identification theory  \cite{ljung2020shift} by addressing issues of bias-variance trade-off, robustness, and model order selection  \cite{pillonetto2014kernel,chiuso2019system,khosravi2021robust},
unifying the identification of continuous-time systems and discrete-time systems  \cite{pillonetto2014kernel}, and allowing the inclusion of various side-information forms in the identification problem 
\cite{fujimoto2017extension,zheng2021bayesian,darwish2018quest,khosravi2019positive,prando2017maximum,fujimoto2018kernel,risuleo2019bayesian,everitt2018empirical,risuleo2017nonparametric,khosravi2021POS,khosravi2021SSG,khosravi2021FDI}. 
Furthermore, due to the inherent connection between RKHSs and Gaussian processes  \cite{wahba1990spline}, kernel-based methods offer a Bayesian interpretation of the system identification problem that allows quantifying the uncertainty and provides statistical guarantees  \cite{lataire2016transfer}.
Over the past decade, research on kernel-based system identification methods has received considerable attention and progressed significantly; nonetheless, it is still an ongoing field of research with various open problems and state-of-the-art results
 \cite{scandella2021kernel,pillonetto2019stable,bisiacco2020mathematical,pillonetto2021sample,bisiacco2020kernel,khosravi2022Lut}.

The building block of each RKHS is the associated kernel function. As a result, various attributes of the RKHS elements are inherited from the corresponding kernel. Therefore, it is necessary to introduce kernels suitable for system identification   \cite{dinuzzo2015kernels}. 
The most prevalent kernels in the literature include diagonal/correlated, tuned/correlated, stable spline, and their extensions, which are proposed primarily for the sake of impulse response stability and smoothness   \cite{zorzi2021second,chen2018continuous,andersen2020smoothing}. 
For improving the identification performance of complex systems, various ideas on designing kernels by combining multiple kernels are proposed  \cite{chen2014system,hong2018multiple-SURE,khosravi2020low,khosravi2020regularized}.
Influenced by machine learning, harmonic analysis of stochastic processes, linear system theory, and filter design techniques, further categories of kernels are developed \cite{chen2018kernel,zorzi2018harmonic,marconato2016filter}.
The significance of kernels led to the investigation of their more generic aspects,
e.g., the relation between the absolute summability of kernels and their stability is clarified in  \cite{bisiacco2020kernel}.  
Moreover, the link between various categories of kernels is studied in \cite{bisiacco2020mathematical}, where 
the mathematical foundations of stable kernels and their RKHSs are explored.
Furthermore, in  \cite{chiuso2019system}, it is shown that the realizations of a zero-mean Gaussian process are almost surely stable impulse responses if the corresponding kernel is diagonally square root integrable (DSRI).

In this work, we revisit the definition and notion of DSRI\footnote{Throughout this paper, DSRI stands for both of ``diagonally square root integrable'' and ``diagonally square root integrability''.} kernels, which was initially introduced in  \cite{chiuso2019system}.
Following this, we investigate the class of DSRI kernels by describing its structure as a partially ordered cone.
We show that this kernel category includes a broad range of well-known kernels commonly used in system identification, e.g., diagonally/correlated, stable spline, amplitude-modulated locally stationary, and simulation-induced kernels. The structure of DSRI kernel class is further elaborated by revisiting the fact that they are stable and integrable. This way, we obtain inner and outer approximations for the class of DSRI kernels.
Subsequently, we investigate fundamental topological features of RKHSs with DSRI kernels. Namely, it is shown that for linear operators defined on $\Lscrone$, the space of stable impulse responses, the continuity property is inherited when the operator is restricted to a RKHS endowed with a DSRI kernel.
For the stability of zero-mean Gaussian processes, we show that the sufficient condition introduced in  \cite{chiuso2019system} is also necessary.
We further generalize this result and provide suitable interpretations. 
Due to the theoretical nature of the work and in an effort to further facilitate reading the manuscript, the burdensome technical arguments, such as proofs of theorems and lemmas, have been moved to the appendix.
For the sake of completeness, the appendix provides all of the proofs, including the relatively simple ones. 

\section{Notation and Preliminaries} 
Throughout the paper, 
the set of natural numbers, the set of real numbers, the set of complex numbers, the set of non-negative integers, and the set of non-negative real numbers are denoted respectively by $\Nbb$, $\Rbb$, $\Cbb$, $\Zbb_+$,  and  $\Rbb_+$.
Moreover, $\Tbb$ denotes the time index set, which corresponds to either  to $\Zbb_+$ or $\Rbb_+$, and $\Tbb_{\pm}$ is defined as $\Tbb_{\pm}:=\Tbb\cup(-\Tbb)$.  
The generic measure space in our discussion is $(\Tbb,\GT,\mu)$, where $\GT$ and $\mu$ are respectively the $\sigma$-algebra of Borel subsets of $\Rbb_+$ and the Lebesgue measure, when $\Tbb=\Rbb_+$, and, 
$\GT$ and $\mu$ are respectively 
the set of subsets of $\Zbb_+$ and the counting measure, when $\Tbb=\Zbb_+$.
Accordingly, 
we additionally consider the measure space $(\Tbb\times\Tbb,\GT\otimes\GT,\mu\times\mu)$, where $\GT\otimes\GT$ and $\mu\times\mu$ are respectively the product $\sigma$-algebra and product measure defined based on $\GT$ and $\mu$.
Furthermore, we assume $\Rbb$ is endowed with Borel  $\sigma$-algebra $\Bscr$ and Lebesgue measure.
Given a measurable space $(\Xcal,\Fscr)$, the space of measurable functions $\vcv:\Xcal\to \Rbb$ is denoted by $\Rbb^{\Xcal}$, and 
$\vcv\in\Rbb^{\Xcal}$ is shown entry-wise as $\vcv=(v_x)_{x\in\Xcal}$, or $\vcv=(v(x))_{x\in\Xcal}$.
Given $\Ycal\subset\Xcal$, the \emph{indicator} function
$\onefun_{\Ycal}:\Xcal\to\{0,1\}$ is defined as
$\onefun_{\Ycal}(x) = 1$, if $x\in\Ycal$, and $\onefun_{\Ycal}(x) = 0$, otherwise.
Depending on the context,  $\Lscrinfty$ denotes  $\ellinfty(\Zbb)$ or $\Linfty(\Rbb)$. 
Similarly, $\Lscrone$ refers to  $\ellone(\Zbb_+)$ or $\Lone(\Rbb_+)$. 
For $p\in\{1,\infty\}$, the norm in $\Lscrp$ is denoted by $\|\cdot\|_{p}$.
The norms defined on Banach spaces $\Lscrone$  and $\Lscrinfty$
are respectively denoted by $\|\cdot\|_{1}$ and $\|\cdot\|_{\infty}$.
The space of bounded linear operators from Banach space $\Xbb$ to Banach space $\Ybb$ is a Banach space, denoted by $\Lcal(\Xbb,\Ybb)$ and endowed with norm $\|\cdot\|_{\Lcal(\Xbb,\Ybb)}$ \cite{brezis2011functional}. 

\section{Diagonally Square Root Integrable Kernels} \label{sec:DSRI}
In this section, the definition of diagonally square root integrable kernels is revisited. To this end, we need to recall the notion of Mercer kernels \cite{berlinet2011reproducing}.
\begin{definition}[\cite{berlinet2011reproducing}]  
	\label{def:kernel_and_section}
	The symmetric measurable function $\kernel:\Tbb\times\Tbb\to \Rbb$ is said to be a {\em positive-definite kernel}, or simply, \emph{kernel}, when, for any $m\in\Nbb$, $s_1,\ldots,s_n\in\Tbb$, and $a_1,\ldots,a_n\in\Rbb$, we have  
	$\sum_{i,j=1}^{m}a_i\kernel(s_i,s_j)a_j\ge 0$.
	For each $t\in\Tbb$, the function $\kernel_t:\Tbb\to\Rbb$, defined as $\kernel_{t}(\cdot)=\kernel(t,\cdot)$, is called the {\em section} of kernel $\kernel$ at $t$.
\end{definition}  
The following definition introduces our main object of interest in this paper.
\begin{definition}\label{def:DSRI_kernel}
	The positive-definite kernel $\,  \kernel:\Tbb\times\Tbb\to \Rbb$ is said to be  \emph{diagonally square root integrable} (DSRI)
	if $\DSRIvalue(\kernel)<\infty$, where $\DSRIvalue(\kernel)$ is defined as 
	\begin{equation}\label{eqn:DSRI_R_+_Z_+}
		\ra{1.5}
		\DSRIvalue(\kernel):=
		\left\{
		\begin{array}{ll} 
			\displaystyle\int_{\Rbb_+}\kernel(t,t)^{\frac12}\  \drm t,
			& 
			\text{\quad  when } \Tbb=\Rbb_+,\\	
			\displaystyle\sum_{t\in\Zbb_+}\kernel(t,t)^{\frac12},  
			& 
			\text{\quad when } \Tbb=\Zbb_+.\\
		\end{array}
		\right.
	\end{equation}
	The class of DSRI kernels is denoted by $\DSRIclass$.
\end{definition}
For any $t\in\Tbb$, one should note that $\kernel(t,t)\ge 0$, which is implied by positive-definiteness property given in Definition~\ref{def:kernel_and_section}.  
Consequently, the right-hand sides in \eqref{eqn:DSRI_R_+_Z_+} are well-defined for any positive-definite kernel, with possible values in $\Rbb_+\cup\{+\infty\}$. According to Definition~\ref{def:DSRI_kernel}, kernel $\kernel$ is DSRI when this value is finite, i.e., $\DSRIvalue(\kernel)<\infty$.   

Given the definition of the DSRI kernels, it is natural to ask about the kernels satisfying this property and their particular features of interest. These questions will be addressed in the following sections.

\section{Well-known DSRI Kernels} \label{sec:DSRI_well_known_kernels}
In this section, we study the class of DSRI kernels, $\DSRIclass$, by showing that many well-known kernels in the system identification context belong to this category of kernels.
To this end, we need the notion of (diagonal) \emph{dominancy}, which introduces a partial order on the set of  positive-definite kernels. 
\begin{definition}\label{def:kerenl_dominance}
	Let $\kernel,\kernelh:\Tbb\times\Tbb\to\Rbb$ be  positive-definite kernels.
	We say $\kernelh$  \emph{dominates} $\kernel$ if there exists $C\in\Rbb_+$ such that  $|\kernel(s,t)|\le C|\kernelh(s,t)|$, for all $t,s\in\Tbb$. 
	Similarly, it is said that $\kernelh$ \emph{diagonally dominates} $\kernel$ if the inequality holds when $s$ equals $t$.
\end{definition}
To elaborate on the importance of Definition~\ref{def:kerenl_dominance}
in describing $\DSRIclass$, we need
to introduce
\emph{finite-rank exponential kernels}.
More precisely, given $n\in\Nbb$, $\bflambda = [\lambda_1,\ldots,\lambda_n]^\tr\in\Rbb_+^n$, and $\bfalpha = [\alpha_1,\ldots,\alpha_n]^\tr\in[0,1)^n$, 
the 
\emph{rank-$n$ exponential} kernel $\kernelRnE{n}:\Tbb\times\Tbb\to\Rbb$ is defined as
\begin{equation}\label{eqn:RnE_kernel}
    \kernelRnE{n}(s,t) 
    = 
    \sum_{i=1}^n \lambda_i \alpha_i^{\frac12(s+t)},    
\end{equation}
for any $s,t\in \Tbb$. 
We denote the kernel by $\kernelRnE{n}(\cdot,\cdot\,;\bflambda,\bfalpha)$, 
and write $\kernelRnE{n}(s,t\,;\bflambda,\bfalpha)$ 
on the left-hand side of \eqref{eqn:RnE_kernel}, when we want to highlight the dependency  on the \emph{hyperparameter} vectors $\bflambda$ and $\bfalpha$.
\begin{theorem}\label{thm:dominancy-and-k_RnE_DSRI}
	\emph{i)} Let $\kernel,\kernelh:\Tbb\times\Tbb\to\Rbb$ be positive-definite kernels where 
	$\kernelh$ is DSRI. If $\kernelh$ (diagonally) dominates $\kernel$, then $\kernel$ is DSRI.\\
	\emph{ii)} The rank-$n$ exponential kernel $\kernelRnE{n}:\Tbb\times\Tbb\to\Rbb$ defined in \eqref{eqn:RnE_kernel} is DSRI.
\end{theorem}
Theorem~\ref{thm:dominancy-and-k_RnE_DSRI} can be used to show that a variety of kernels belongs to $\DSRIclass$.
In the literature of system identification, various kernels are introduced  \cite{pillonetto2014kernel,chen2012estimation}, e.g., 
\emph{diagonal},
\emph{diagonally/correlated},
\emph{tuned/correlated},
and 
\emph{stable spline} kernels,
which are respectively denoted by 
$\kernelDI$, $\kernelDC$, $\kernelTC$, and $\kernelSS$,
and defined as
\begin{align}
    &\kernelDI(s,t) = 
    \onefun_{\{0\}}(s-t)  \alpha^{s},
    \label{eqn:DI_kernel}
    \\
    &\kernelDC(s,t)=\alpha^{\frac12(s+t)} \gamma^{|s-t|},
    \label{eqn:DC_kernel}
    \\
    &\kernelTC(s,t)=\alpha^{\max(s,t)}, 
    \label{eqn:TC_kernel}
    \\
    &\kernelSS(s,t)=\alpha^{\max(s,t)+s+t}-\frac{1}{3}\alpha^{3\max(s,t)},
    \label{eqn:SS_kernel}
\end{align}
for any $s,t\in\Tbb$, 
where $\alpha\in(0,1)$, $\gamma\in (-1,1)$, if $\Tbb=\Zbb_+$, and, $\gamma\in (0,1)$, if $\Tbb=\Rbb_+$. 
Moreover, in \cite{pillonetto2016AtomicNuclearKernel}, 
the first and second order \emph{integral stable spline}  kernels
are defined as 
\begin{align}
    &\kerneliTC(s,t)
    = 
    \frac{\alpha^{\max(s,t)+1}-\beta^{\max(s,t)+1}}{\max(s,t)+1},
    \label{eqn:iTC_kernel}
    \\
    \begin{split}
        &\kerneliSS(s,t)
        =
        \frac{\alpha^{s+t+\max(s,t)+1}-\beta^{s+t+\max(s,t)+1}}{s+t+\max(s,t)+1}
        \\
        &\qquad\qquad\qquad
        -\frac{\alpha^{3\max(s,t)+1}-\beta^{3\max(s,t)+1}}{9\max(s,t)+3},
    \end{split}
    \label{eqn:iSS_kernel}
\end{align}
for any $s,t \in\Tbb$, where $0\le \beta \le \alpha<1$.
We can directly calculate $\DSRIvalue(\kernel)$ using \eqref{eqn:DSRI_R_+_Z_+},
for the above-mentioned kernels, and show that these kernels belong to $\DSRIclass$. 
On the other hand, we can easily see that
kernels $\kernelDI$, $\kernelDC$, $\kernelTC$, and $\kerneliTC$ are dominated by 
$\kernelRnE{n}(\cdot,\cdot\,;1,\alpha)$.
Similarly, we can show that the 
$\kernelRnE{n}(\cdot,\cdot\,;1,\alpha^3)$
dominates $\kernelSS$ and $\kerneliSS$.
Thus, one can easily conclude from Theorem~\ref{thm:dominancy-and-k_RnE_DSRI} that 
each of the above-mentioned kernels are DSRI.
Based on the same line of argument, one can show the same result for the $n^\nth$-order stable spline kernels \cite{pillonetto2014kernel} (see Appendix~\ref{apn:DSRI_SSn} for more details).
\begin{theorem}\label{thm:sum_product_DSRI}
	Let $\kernel,\kernelh:\Tbb\times\Tbb\to\Rbb$ be  positive-definite kernels, where $\kernel$ is DSRI. \\	
	\emph{i)} If $\kernelh$ is DSRI, then $\alpha\kernel+\beta\kernelh$ is a DSRI kernel, for any $\alpha,\beta\in\Rbb_+$.\\
	\emph{ii)} If $\sup_{t\in\Tbb}\kernelh(t,t)<\infty$, then $\kernel\kernelh$ is a DSRI kernel.
\end{theorem}

Theorem~\ref{thm:dominancy-and-k_RnE_DSRI} and Theorem~\ref{thm:sum_product_DSRI} characterize the structure of the class of DSRI kernels as a cone equipped with a partial order. Also, they can further be used to verify the DSRI property for other kernels.
For example, consider
kernel $\kerneliTS$ introduced in \cite{pillonetto2016AtomicNuclearKernel} as the combination of $\kerneliTC$ and $\kerneliSS$, i.e., we have 
$\kerneliTS(s,t) := \kerneliTC(s,t) +\kerneliSS(s,t)$,   
for any $s,t \in\Tbb$.
Based on the above discussion and Theorem~\ref{thm:sum_product_DSRI}, one can easily see that $\kerneliTS$ is a DSRI kernel.

Let $\vcv :=(v_t)_{t\in\Tbb}\in\Lscrone$ and $\kernel_{\vcv}:\Tbb\times\Tbb\to\Rbb$ be 
defined as
$\kernel_{\vcb}(s,t)=v_sv_t$, 
for any $s,t\in\Tbb$ \cite{chen2018kernel}.
One can easily see that $\kernel_{\vcv}$ is a rank-$1$  positive-definite kernel with
\begin{equation*}\!\!\!\!
\DSRIvalue(\kernel_{\vcv}) 
=
\begin{cases}
\displaystyle\int_{\Rbb_+}(v_t^2)^{\frac12}\drm t 
= 
\displaystyle\int_{\Rbb_+}|v_t|\drm t,
&\text{ if }\Tbb=\Rbb_+,\\
\displaystyle\sum_{t\in\Zbb_+}(v_t^2)^{\frac12}
=
\displaystyle\sum_{t\in\Zbb_+}|v_t|, 
&\text{ if }\Tbb=\Zbb_+,
\end{cases}
\end{equation*}
which says that $\DSRIvalue(\kernel_{\vcv})=\|\vcv\|_1$.
This implies that $\kernel_{\vcv}\in\DSRIclass$. 
In~\cite{chen2018kernel}, the \emph{amplitude modulated locally stationary} (AMLS) kernels are introduced, which  are generalized form of $\kernel_{\vc{v}}$. 
More precisely, let $\kernelst:\Tbb\times\Tbb\to \Rbb$ be a stationary  positive-definite kernel, i.e., we have 
\begin{equation}
\kernelst(s+\tau,t+\tau)=\kernelst(s,t), 
\qquad \forall s,t,\tau\in\Tbb.
\end{equation}
Subsequently, the AMLS kernel $\kernelAMLS:\Tbb\times\Tbb\to \Rbb$ is defined as
\begin{equation}\label{eqn:kernel_AMLS}
\kernelAMLS(s,t) = v_t \kernelst(s,t)v_s,   
\qquad \forall s,t\in\Tbb.
\end{equation}
Note that since $\kernelst$ is a stationary kernel, 
we know that $\sup_{t\in\Tbb}\kernelst(t,t) = \kernelst(0,0)<\infty$.
Therefore, due to Theorem~\ref{thm:sum_product_DSRI} and $\kernel_{\vcb}\in\DSRIclass$, we have  $\kernelAMLS\in\DSRIclass$. 
In addition to $\kernelAMLS$, the \emph{simulation induced} kernels are introduced in \cite{chen2018kernel}.
Similar to our previous discussion, one can show that under certain conditions, the simulation induced kernels are DSRI (see Appendix \ref{apn:DSRI_SI} for more details).

We can show that the DSRI property is preserved under proper  sampling (see Appendix~\ref{apn:DSRI_sampling}) and  reparameterization of the arguments of the kernel (see Appendix~\ref{apn:DSRI_reparameterization}).
Using Theorem~\ref{thm:dominancy-and-k_RnE_DSRI} and Theorem~\ref{thm:sum_product_DSRI}, based on the discussion provided in this section, and following line of arguments similar to Appendices~\ref{apn:DSRI_SSn}, \ref{apn:DSRI_SI}, \ref{apn:DSRI_sampling}, and \ref{apn:DSRI_reparameterization},
one can show that a broad range of kernels are DSRI.
The class of DSRI kernels is further studied in the next section.

\section{DSRI Kernels: Stability and Integrability}
\label{sec:DSRI_stability_integrability}
To elaborate further on the structure of the class of DSRI kernels, we investigate their stability and integrability properties in this section. 
Since in the kernel-based system identification framework, the kernel attributes are inherited by the identified model, one may ask about the main feature of concern, which is the stability of the kernel. To address this question, we need to recall the notion of stable kernels \cite{pillonetto2014kernel}.
\begin{definition}[\cite{pillonetto2014kernel}]
	The positive-definite kernel $\,\kernel:\Tbb\times\Tbb\to \Rbb$ is said to be  \emph{stable}  if, for any $\vc{u}=(u_s)_{s\in\Tbb}\in\Lscr^{\infty}$,  
	one has  
	\begin{equation}\label{eqn:stability_kernel_R_+_Z_+}  
		\ra{1.5}
		\!\!\!
		\left\{
		\begin{array}{ll} 
			\displaystyle\int_{\Rbb_+}\Big|\displaystyle\int_{\Rbb_+}u_s\kernel(t,s)\drm s\Big|\drm t<\infty,
			\!& \!
			\text{\quad  when } \Tbb=\Rbb_+,\!\\	
			\displaystyle\sum_{t\in\Zbb_+}\Big|\displaystyle\sum_{s\in\Zbb_+}u_s\kernel(t,s)\Big|<\infty,,  
			\!& \!
			\text{\quad when } \Tbb=\Zbb_+.\!\\
		\end{array}
		\right.
	\end{equation}
	The class of stable kernels is denoted by $\Sstable$.
\end{definition}
The following theorem demonstrates the relationship between the DSRI kernels and the stable kernels.
\begin{theorem}[\cite{chiuso2019system}]\label{thm:DSRI_stable} 
    Every DSRI kernel is stable. 
\end{theorem}
We have already verified that 
$\DSRIclass\subseteq\Sstable$.
In addition to stable kernels, a well-known interesting category of kernels in the context of system identification are the integrable ones. In the following, we review their definition.
\begin{definition}[\cite{pillonetto2014kernel}]\label{def:integrable_kernel}
The positive-definite kernel $\,\kernel:\Tbb\times\Tbb\to \Rbb$ is called \emph{integrable}
if we have
\begin{equation}\label{eqn:abs_summable_integrable}  
	\ra{1.5}
	\!\!\!
	\left\{
	\begin{array}{ll} 
		\displaystyle\int_{\Rbb_+}\displaystyle\int_{\Rbb_+}\big|\kernel(t,s)\big|\drm s \drm t<\infty,
		\!& \!
		\text{\quad  when } \Tbb=\Rbb_+,\!\\	
		\displaystyle\sum_{t\in\Zbb_+}\displaystyle\sum_{s\in\Zbb_+}\big|\kernel(t,s)\big|<\infty,,  
		\!& \!
		\text{\quad when } \Tbb=\Zbb_+.\!\\
	\end{array}
	\right.
\end{equation}
The class of integrable kernels is denoted by $\Sscr_1$.
\end{definition}
It is known that the set of integrable kernels is a subclass of stable kernels \cite{pillonetto2014kernel,bisiacco2020mathematical}, i.e., $\Sscr_1\subseteq\Sstable$.
The following theorem further characterizes the class of DSRI kernels by elaborating their connection with the integrable kernels. 
This theorem
is implicitly implied from the proof of Lemma~2 in \cite{chiuso2019system}.
\begin{theorem}[\cite{chiuso2019system}]\label{thm:DSRI_integrable}
    Every DSRI kernel is integrable. 
\end{theorem}

In \cite{bisiacco2020kernel}, it is verified that there exists a stable kernel $\kernel:\Zbb_+\times\Zbb_+\to \Rbb$ which is not integrable, i.e., $\sum_{s,t\in\Zbb_+}|\kernel(s,t)|=\infty$.
The next theorem verifies a similar property for DSRI kernels.
\begin{theorem}\label{thm:integrable_not_DSRI}
	There exists an integrable kernel which is not a DSRI kernel.
\end{theorem}
The following corollary is a direct result of Theorem~\ref{thm:integrable_not_DSRI} and the fact that any integrable kernel is stable \cite{pillonetto2014kernel}.
\begin{corollary}\label{cor:stable_not_DSRI}
	There exists a stable kernel which is not a DSRI kernel.
\end{corollary}

In \cite{bisiacco2020mathematical}, other categories of positive-definite kernels are considered. 
The positive-definite kernel $\kernel:\Tbb\times\Tbb\to\Rbb$ is said to be \emph{finite-trace} if we have 
\begin{equation}\label{eqn:finite_trace_kernels}
 	\ra{1.5}\left\{
 	\begin{array}{ll} 
		\displaystyle\Sum_{s\in\Zbb_+} \kernel(s,s)<\infty,
		& 
		\text{\quad  if } \Tbb=\Zbb_+,\\	
		\displaystyle\int_{\Rbb_+}\ \kernel(s,s)\drm s <\infty,  
		& 
		\text{\quad if } \Tbb=\Rbb_+.\\
	\end{array}
	\right.
\end{equation}
Similarly, it is called a \emph{squared integrable} kernel if
\begin{equation}\label{eqn:squared_integrable_kernels}
 	\ra{1.5}\left\{
 	\begin{array}{ll} 
		\displaystyle\sum_{s\in\Zbb_+}\displaystyle\sum_{t\in\Zbb_+} \kernel(s,t)^2<\infty,
        & 
        \text{\quad if } \Tbb=\Zbb_+,\\	
		\displaystyle\int_{\Rbb_+}\displaystyle\int_{\Rbb_+}\ \kernel(s,t)^2\drm s\, \drm t <\infty,  
		& 
		\text{\quad if } \Tbb=\Rbb_+.\\
	\end{array}
	\right.
\end{equation}
The class of finite-trace kernels and the class of squared integrable kernels are denoted by $\Sfinitetrace$ and $\Sscr_2$, respectively \cite{bisiacco2020mathematical}.
Based on the above discussion and \cite{bisiacco2020mathematical}, we have
\begin{equation}
    \DSRIclass
    \subset
    \Sscr_1
    \subset
    \Sstable
    \subset
    \Sfinitetrace
    \subset
    \Sscr_2,
\end{equation}
where all of the inclusions are strict. 

See Figure \ref{KRI:fig:inclusion_classes} for an illustration of the discussion presented in the current section and the previous section. One should compare this figure with Figure 1 in \cite{bisiacco2020mathematical}.
\begin{figure}[t]
	\centering
	\includegraphics[width=0.48\textwidth]{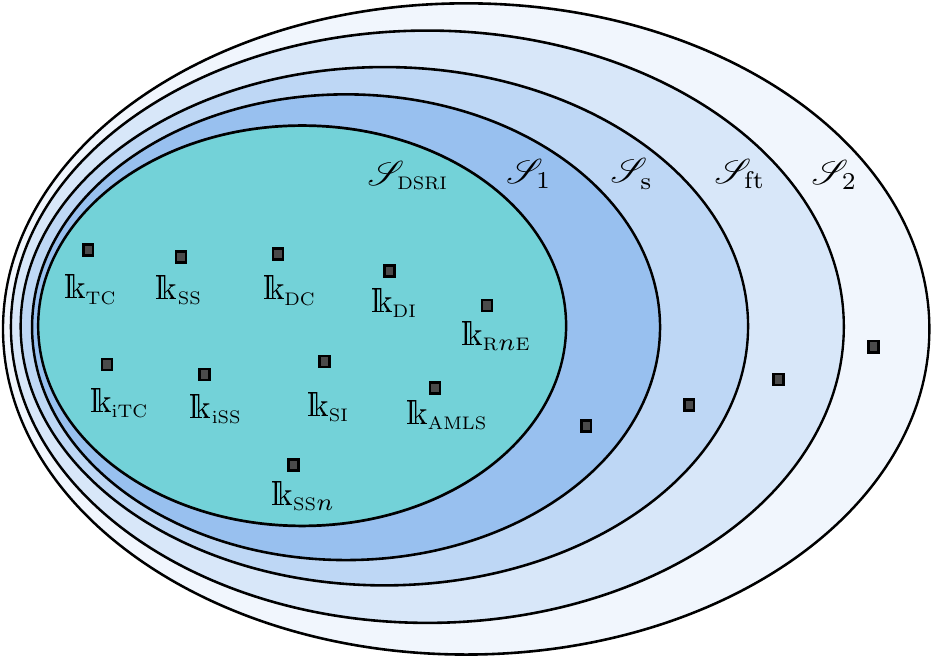}
	\caption{Illustration of the inclusion features for different kernel categories discussed in Section~\ref{sec:DSRI_well_known_kernels} and Section~\ref{sec:DSRI_stability_integrability}}.
	\label{KRI:fig:inclusion_classes}
	\vspace{0mm}
\end{figure}

\section{Operator Continuity and DSRI Kernels}
\label{sec:DSRI_RKHS}
In this section, we study certain topological features of the RKHSs equipped with DSRI kernels, namely the continuity of linear operators defined on them.

We recall 
that with respect to each positive-definite kernel, a Hilbert space is defined uniquely \cite{aronszajn1950theory}.
More precisely, based on the Moore-Aronszajn theorem, these Hilbert spaces are exactly the ones where the evaluation functionals are bounded \cite{aronszajn1950theory,berlinet2011reproducing}.
\begin{theorem}[\cite{berlinet2011reproducing}]
\label{thm:kernel_to_RKHS_def} 
	Given a positive-definite kernel $\,\kernel:\Tbb\times\Tbb\to \Rbb$, there exists a unique Hilbert space $\Hk\subseteq \Rbb^{\Tbb}$ with inner product $\inner{\cdot}{\cdot}_{\Hk}$, referred to as the \emph{RKHS with kernel} $\kernel$, where for each $t\in\Tbb$, we have \\
	\emph{i)} $ \kernel_t\in\Hk$, and\\
	\emph{ii)} $g_t = \inner{\vc{g}}{ \kernel_{t}}_{\Hk}$, for all $\vc{g}=(g_s)_{s\in\Tbb}\in\Hk$.\\
	The second feature is called the {\em reproducing property}.
\end{theorem}

In the context of system identification, the RKHSs endowed with the stable kernels are of special interest due to their particular feature reviewed in the following theorem.
\begin{theorem}[\cite{pillonetto2014kernel,chen2018stability,carmeli2006vector}]\label{thm:kernel_stability}
	Let $\kernel:\Tbb\times\Tbb\to \Rbb$ be a  positive-definite kernel. Then, $\Hk\subseteq\Lscr^1$  if and only if $\kernel$ is a stable kernel. In this case, $\Hk$ is called a \emph{stable RKHS}.
\end{theorem}
Given a stable kernel $\kernel$, we know that $\Hk\subseteq\Lscr^1$.
Accordingly, various objects introduced on $\Lscr^1$ can be redefined by restricting them to $\Hk$. Here, one may ask about the inherited properties followed by this restriction.
The main feature of DSRI kernels is that the continuity of operators defined on $\Lscrone$ is inherited when they are restricted to the corresponding RKHS.
\begin{theorem}\label{thm:continuity_inherits}
	Let $\Banach$ be a Banach space equipped with norm $\|\cdot\|_{\Banach}$ and $\mxL:\Lscrone\to\Banach$ be a continuous operator.
	If $\kernel:\Tbb\times\Tbb$ is a DSRI kernel, then $\Hk$ is a linear subspace of $\Lscrone$ and $\mxL:\Hk\to\Banach$ is continuous.
	Moreover, we have
	\begin{equation}\label{eqn:continuity_inherits_thm:norm_ineq}
		\|\mxL\|_{\Lcal(\Hk,\Banach)}
		\le 
		\|\mxL\|_{\Lcal(\Lscrone,\Banach)} \DSRIvalue(\kernel).
	\end{equation}
\end{theorem} 
Given a Banach space $\Banach$ with norm $\|\cdot\|_{\Banach}$, we denote by $\Lscrinfty(\Tbb;\Banach)$ the space of $\Banach$-valued Bochner measurable functions where the essential supremum of their norm in $\Banach$ is bounded, i.e., for any $\vcv=(v_t)_{t\in\Tbb}\in\Lscrinfty(\Tbb;\Banach)$, we have $\esssup_{t\in\Tbb}\|v_t\|_{\Banach}<\infty$ \cite{mikusinski1978Bochner}.
\begin{theorem}\label{thm:Lv_continuity}
Let $\,\vcv$ be an arbitrary element in $\Lscrinfty(\Tbb;\Banach)$ and $\kernel:\Tbb\times\Tbb$ be a positive-definite kernel. 
Define an operator $\mxL:\Hk\to\Banach$ as follows
\begin{equation}\label{eqn:Lv_operator}
	\mxL(\vcg) := 
	\begin{cases}
		\displaystyle\sum_{s\in \Zbb_+}g_tv_t, 
		& \text{ if } \Tbb=\Zbb_+,\\	
		\displaystyle\int_{\Rbb_+}\!\!g_tv_t \ \!\drm t, 
		& \text{ if } \Tbb=\Rbb_+,\\
	\end{cases}
\end{equation}
for any $\vcg=(g_t)_{t\in\Tbb}\in\Hk$.
If $\kernel$ is a DSRI kernel, then $\mxL$ is a continuous linear operator.
\end{theorem}

Theorem \ref{thm:continuity_inherits} and Theorem \ref{thm:Lv_continuity} allow one to transfer different existing results for BIBO stable impulse responses to RKHS $\Hk$. The following corollaries are examples of this.
\begin{corollary}\label{cor:Lut_bounded}
	Let $\kernel$ be a DSRI kernel, $\vcu\in\Lscrinfty$ be a bounded signal and $t\in\Tbb_{\pm}$. 
	Define the \emph{convolution} operator $\Lu{t}:\Hk\to\Rbb$ as
    \begin{equation}\label{eqn:Lu_operator}
	\Lu{t}(\vc{g}) := 
	\begin{cases}
		\displaystyle\sum_{s\in \Zbb_+}g_s u_{t-s}, 
		& \text{ if } \Tbb=\Zbb_+,\\	
		\displaystyle\int_{\Rbb_+} g_s u_{t-s}\drm s,  
		& \text{ if } \Tbb=\Rbb_+,\\
	\end{cases}
    \end{equation}
    for any $\vc{g}=(g_t)_{t\in\Tbb} \in\Hk$.
    Then, $\Lu{t}:\Hk\to\Rbb$ is a continuous linear operator.
\end{corollary}
Let $\Omega_{\Tbb}$ be defiend as $\Omega_{\Tbb}:=[0,\pi]$ when $\Tbb=\Zbb_+$, and $\Omega_{\Tbb}:=\Rbb_+$ when $\Tbb=\Rbb_+$.
With respect to each $\omega$ in $\Omega_{\Tbb}$, the operators $\Fwr:\Hk\to\Rbb$ and  $\Fwi:\Hk\to\Rbb$ are defined
respectively as
\begin{equation}\label{eqn:Fwr_operator}
	\Fwr(\vc{g}) := 
	\begin{cases}
		\displaystyle\sum_{t\in\Zbb_+} g_t \cos(\omega t), 
		& \text{ if } \Tbb=\Zbb_+,\\	
		\displaystyle\int_{\Rbb_+} g_t\cos(\omega t)\drm t,  
		& \text{ if } \Tbb=\Rbb_+,\\
	\end{cases}
\end{equation}
and
\begin{equation}\label{eqn:Fwi_operator}
	\Fwi(\vc{g}) := 
	\begin{cases}
		- \displaystyle\sum_{t\in\Zbb_+} g_t \sin(\omega t),
        & \text{ if } \Tbb=\Zbb_+,\\	
		- \displaystyle\int_{\Rbb_+}  g_t\sin(\omega t)\drm t,  
		& \text{ if } \Tbb=\Rbb_+,\\
	\end{cases}
\end{equation}
for any $\vc{g}=(g_t)_{t\in\Tbb} \in\Hk$.
Moreover, we define $\Fw:\Hk\to\Cbb$ as $\Fw = \Fwr + \Jimage\Fwi$, where $\Jimage$ denotes imaginary unit.  
One can see that $\Fwi(\vc{g})$ and $\Fwi(\vc{g})$ respectively corresponds to the real and imaginary part of Fourier transform of impulse response $\vcg\in\Hk$ evaluated at frequency $\omega\in\Omega_{\Tbb}$, which is $\Fw(\vc{g})$ .
From Theorem~\ref{thm:Lv_continuity}, we have the following corollary for the introduced operators.
\begin{corollary}\label{cor:Fwr_Fwi_bounded}
	Let $\kernel$ be a DSRI kernel.
    Then, $\Fwr$, $\Fwi$ and $\Fw$ are continuous linear operators, for all $\omega\in\Omega_{\Tbb}$.
\end{corollary}

\section{Stable Gaussian Processes}
Let $(\Omega,\GOmega,\Pbb)$ be a probability space, where $\Omega$  is the sample space, $\GOmega$ is the corresponding $\sigma$-algebra, and $\Pbb$ is the probability measure defined on $\GOmega$. 
Given a measurable function $\vcm=(m_t)_{t\in\Tbb}$ and a positive-definite kernel $\kernel:\Tbb\times\Tbb\to\Rbb$, the stochastic process
\begin{equation}
    \vcg:(\Tbb\times\Omega,\GT\otimes\GOmega,\mu\times\Pbb)\to\Rbb
\end{equation}
is called a \emph{Gaussian process} (GP) with mean $\vcm$ and kernel $\kernel$ \cite{berlinet2011reproducing}, denoted by $\GP(\vcm,\kernel)$, when, for  any $n\in\Nbb$ and any $t_1,\ldots,t_n\in\Tbb$, the random vector $[g_{t_1},\ldots,g_{t_n}]^\tr$ has a Gaussian distribution as follows 
\begin{equation}
    \big[g_{t_1},\ldots,g_{t_n}\big]^\tr
    \sim
    \Ncal\Big(\big[m_{t_i}\big]_{i=1}^n,\big[\kernel(t_i,t_j)\big]_{i,j=1}^n\Big).
\end{equation}
The following definition reviews the notion of an interesting class of Gaussian processes in the context of system identification \cite{chiuso2019system}. 
\begin{definition}[\cite{chiuso2019system}]
The Gaussian process $\GP(\vcm,\kernel)$ is said to be \emph{stable} in the BIBO sense if its realizations, also known as  sample paths, are almost surely BIBO stable impulse responses, i.e., $\Pbb[\|\vcg\|_1<\infty]=1$.
\end{definition} 
The importance of stable GPs is according to their role in the Bayesian interpretation of kernel-based impulse response identification. 
Hence, one may ask about the necessary and sufficient conditions for the stability of the Gaussian process $\GP(\vcm,\kernel)$ \footnote{This question has been raised  during workshop ``\emph{Bayesian and Kernel-Based Methods in Learning Dynamical Systems}'', 21${}^\text{\tiny{st}}$ IFAC World Congress, Berlin, Germany, 2020.}.
Part of this question is addressed in \cite{chiuso2019system}, which is reviewed in the following lemma.
\begin{lemma}[\cite{chiuso2019system}]\label{lem:DSRI_GP_stable}
Let $\kernel$ be a positive-definite kernel and $\vcg\sim\GP(\zero,\kernel)$, where $\zero$ denotes the constant zero function. 
If kernel $\kernel$ is DSRI, then we have $\Pbb[\|\vcg\|_1<\infty]=1$. 
\end{lemma}
According to Lemma~\ref{lem:DSRI_GP_stable}, the DSRI feature of $\kernel$ is a sufficient condition for the almost sure BIBO stability of $\vcg$ when $\vcg\sim\GP(\zero,\kernel)$.
The following lemma concerns the other direction of Lemma~\ref{lem:DSRI_GP_stable}.
Before proceeding further, we need to present additional definitions. 
Let the function $\Phi:\Rbb_+\to[0,1]$ be defined as
\begin{equation}\label{eqn:Phi_fun}
	\Phi(\delta)
	=
	\frac{1}{(2\pi)^{\frac12}}
	\int_{-\delta}^{\delta}\expe^{-\frac12 x^2}\drm x,
\end{equation} 
for any $\delta\in\Rbb_+$.
Note that $\Phi$ is closely related to the Gaussian error function, i.e., $\Phi(\delta)$ is the probability that the value of a standard Gaussian random variable is in the interval $[-\delta,\delta]$, for any $\delta\in\Rbb_+$.
Moreover, one can see that $\Phi$ is a strictly increasing bijective function, and therefore, it has a well-defined inverse $\Phi^{-1}: [0,1]\to \Rbb_+$, which is also a strictly increasing bijective map.
\begin{lemma}\label{lem:GP_stable}
Let $\kernel$ be a positive-definite kernel and $\vcg\sim\GP(\zero,\kernel)$, where $\zero$ is the constant zero function. 
If $\Pbb[\|\vcg\|_1<\infty]>0$, then $\kernel$ is a DSRI kernel and we have 
$\Pbb[\|\vcg\|_1<\infty]=1$.
\end{lemma}
Following this, we have the main theorem of this section which is implied from Lemma~\ref{lem:DSRI_GP_stable} and Lemma~\ref{lem:GP_stable}.
\begin{theorem}\label{thm:GP_stable}
	Let $\vcm=(m_t)_{t\in\Tbb}$ be a stable impulse response and $\kernel:\Tbb\times\Tbb\to\Rbb$ be a positive-definite kernel. Also, let  $\GP(\vcm,\kernel)$ be the Gaussian process with mean impulse response $\vcm$ and kernel $\kernel$. 
	Then, if $\,\kernel$ is a DSRI kernel, we have $\Pbb[\|\vcg\|_1<\infty]=1$, and if $\,\kernel$ is not a DSRI kernel, we have $\Pbb[\|\vcg\|_1<\infty]=0$. 
\end{theorem}
The following corollary is a direct result of Theorem~\ref{thm:GP_stable} and the definition of (BIBO) stability for the Gaussian processes.
\begin{corollary}
Let the assumptions of Theorem~\ref{thm:GP_stable} holds. Then, $\GP(\vcm,\kernel)$ is stable if and only if $\,\kernel$ is a DSRI kernel.
\end{corollary}

The theorem and corollary presented here have an interesting interpretation.
For $\vcg=(g_t)_{t\in\Tbb}\sim\GP(\vcm,\kernel)$ and $t\in\Tbb$, we know that $g_t$ is a random variable with Gaussian distribution $\Ncal(m_t,\kernel(t,t))$. 
Accordingly, with respect to each $\epsilon\in(0,1)$,  we can characterize an $\epsilon$ \emph{confidence interval} based on the standard deviation of $g_t$.
More precisely, the $\epsilon$ confidence interval for $g_t$, denoted by $I_{t,\epsilon}$, is defined as
\begin{equation}
 \Ical_{t,\epsilon}
 =
 [m_t-\delta_{\epsilon} \kernel(t,t)^{\frac12},m_t+ \delta_{\epsilon} \kernel(t,t)^{\frac12}],   
\end{equation}
where $\delta_{\epsilon}$ is the positive real scalar specified as $\delta_{\epsilon}=\Phi^{-1}(\epsilon)$. 
Furthermore, let impulse responses $\vcs_{\epsilon}^+$ and $\vcs_{\epsilon}^-$ be defined respectively as
\begin{equation}
\vcs_{\epsilon}^+:=\big(m_t+ \delta_{\epsilon} \kernel(t,t)^{\frac12}\big)_{t\in\Tbb}\, ,    
\end{equation}
and
\begin{equation}
\vcs_{\epsilon}^-:=\big(m_t - \delta_{\epsilon} \kernel(t,t)^{\frac12}\big)_{t\in\Tbb}\, .
\end{equation}
We know that $\vcs_{\epsilon}^+$ and $\vcs_{\epsilon}^-$ corresponds respectively to the upper and lower bounds of the introduced point-wise $\epsilon$ confidence intervals.
Accordingly, we can define an $\epsilon$ \emph{confidence region}, denoted by $\Rcal_{\epsilon}$, as the union of $\epsilon$ confidence intervals $\{\Ical_{t,\epsilon}\,|\,t\in\Tbb\}$, 
i.e., $\Rcal_{\epsilon} = \cup_{t\in\Tbb}\Ical_{t,\epsilon}$.
One can easily see that $\Rcal_{\epsilon}$ is the region between the impulse responses $\vcs_{\epsilon}^+$ and  $\vcs_{\epsilon}^-$ (see Figure \ref{KRI:fig:DSRI}).
Note that due to the definition of $ \Ical_{t,\epsilon}$, we have $\Pbb\big[ g_t\in  \Ical_{t,\epsilon}\big]=\epsilon$, for any $t\in\Tbb$.
However, one should note that this argument does not imply  $\Pbb\big[ \vcg\in  \Rcal_{\epsilon}\big]\ge\epsilon$.
On the other hand, the theorem and corollary say that $\vcg$ is a stable impulse response with probability one, that is $\Pbb[\|\vcg\|_1<\infty]=1$, if and only if, 
the confidence bound impulse responses $\vcs_{\epsilon}^+$ and  $\vcs_{\epsilon}^-$ are stable, or equivalently, the  $\epsilon$ confidence region $\Rcal_{\epsilon}$ has finite area.
Moreover, if the area of $\Rcal_{\epsilon}$ is infinite, then $\vcg$ is an unstable impulse response with probability one, i.e., $\Pbb[\|\vcg\|_1=\infty]=1$.
In Figure \ref{KRI:fig:DSRI}, we have shown  $50$ sample paths of an example Gaussian process, the corresponding mean impulse response $\vcm$, and the confidence bound impulse responses $\vcs_{\epsilon}^+$ and $\vcs_{\epsilon}^-$, where $\epsilon=0.95$. 
\begin{figure}[t]
	\centering
	\includegraphics[width=0.47\textwidth]{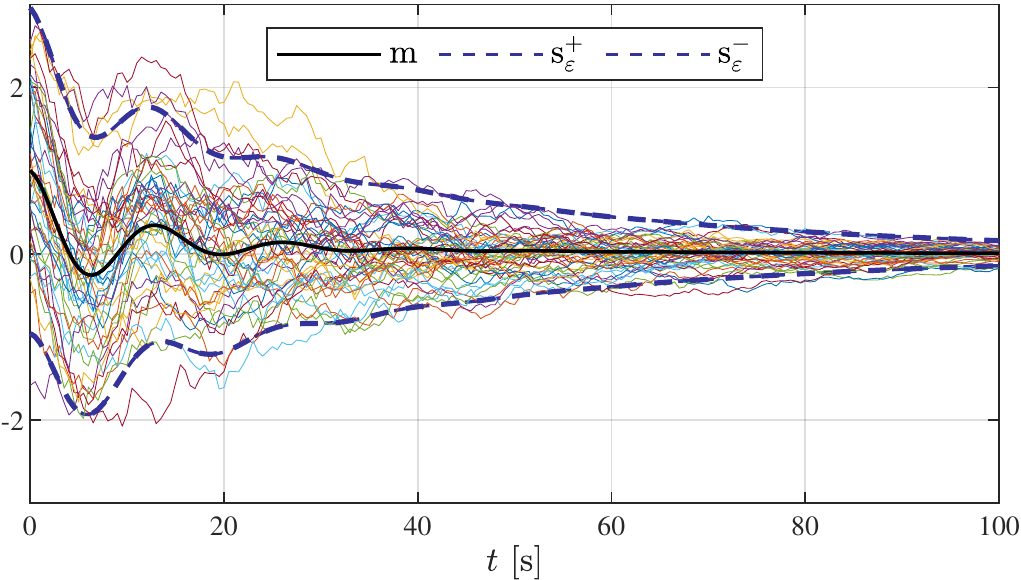}
	\caption{The figure illustrates $50$ sample paths of a Gaussian process, the mean impulse response $\vcm$, and, the $95$\% confidence region between $\vcs_{\epsilon}^+$ and $\vcs_{\epsilon}^-$.}
	\label{KRI:fig:DSRI}
	\vspace{0mm}
\end{figure}

\section{Conclusion}\label{sec:conclusion}
We have investigated the class of diagonally square root integrable 
kernels in this work. It is verified that the category of DSRI kernels includes well-known kernels used in system identification, such as diagonally/correlated, tuned/correlated, stable spline, amplitude-modulated locally stationary, and simulation-induced kernels. 
We have observed that the DSRI kernel category has a cone structure endowed with a partial order. Moreover, this kernel class is a subclass of stable kernels and integrable kernels. 
We have looked into certain fundamental topological properties of the RKHSs with DSRI kernels. More precisely, we have noticed that the continuity of linear operators defined on $\Lscrone$ is inherited when they are restricted to a RKHS equipped with a DSRI kernel.
Furthermore, it has been verified that the realizations of a Gaussian process centered at a stable impulse response are almost surely stable if and only if the corresponding kernel admits the DSRI property.

\appendix
\section{Appendix}\vspace{-1mm} 
\label{sec:AppendixProofs}
\subsection{Proof of Theorem~\ref{thm:dominancy-and-k_RnE_DSRI}} \label{apn:proof:thm:dominancy-and-k_RnE_DSRI}
Part i)
For the case of $\Tbb=\Rbb_+$, one can easily see that
\begin{equation}
	\DSRIvalue(\kernel)
	\le 
	C^{\frac12} 
	\int_{\Rbb_+}\kernelh(t,t)^{\frac12}\,  \drm t 
	= C^{\frac12} \DSRIvalue(\kernelh)<\infty.
\end{equation}
A similar argument holds when $\Tbb=\Zbb_+$.\\
Part ii)
For any  $t_1,\ldots,t_m\in\Tbb$ and $a_1,\ldots,a_m\in\Rbb$, we have
\begin{equation*}	
	\begin{split}
	\sum_{j,k=1}^m a_j\kernelRnE{n}(t_j,t_k)a_k
	&=
	\sum_{j,k=1}^m \sum_{i=1}^n a_ja_k\lambda_i\alpha_i^{\frac12(t_j+t_k)}
	\\
	&=
	\sum_{i=1}^n \lambda_i\Big(\sum_{j=1}^ma_j\alpha_i^{\frac12t_j}\Big)^2
	\ge 0,
	\end{split}
\end{equation*}
which says that $\kernelRnE{n}$ is a positive-definite kernel.
For any $s,t\in\Tbb$, one can see that 
$\kernelRnE{n}(s,t) \le 
\alpha^{\frac12(s+t)}\lambda$,
where $\alpha :=\max_{1\le i \le n}\alpha_i$ and $\lambda:=\sum_{i=1}^n\lambda_i$.
Therefore,  we have
\begin{equation*}
		\DSRIvalue(\kernelRnE{n})
		\le
		\begin{cases}
		\lambda^{\frac12}
		\!
		\displaystyle\int_{\Rbb_+}
		\!\!\!\! \alpha^{\frac12t}\ \! \drm t
		=
		-\frac{2\lambda^{\frac12}}{\ln(\alpha)},
		&
		\text{ if } \Tbb=\Rbb_+,
		\\	
		\lambda^{\frac12}
		\displaystyle\sum_{t\in\Zbb_+}\alpha^{\frac12t}
		=
		\frac{\lambda^{\frac12}}
		{1-\alpha^{\frac12}},
		&
		\text{ if } \Tbb=\Zbb_+,
	\end{cases}
\end{equation*}
which implies that 
$\kernelRnE{n}$ is a DSRI kernel. 
\qed\vspace{-2mm}
\subsection{Proof of Theorem~\ref{thm:sum_product_DSRI}} \label{apn:proof:thm:sum_product_DSRI}
Part i) One can easily see that
\begin{equation}
	\big(\alpha\kernel(t,t)+\beta\kernelh(t,t)\big)^{\frac12}
	\le
	\alpha^{\frac12}\kernel(t,t)^{\frac12}+\beta^{\frac12}\kernelh(t,t)^{\frac12}, 
\end{equation}
for any $t\in\Tbb$.
Accordingly, the proof follows directly from the triangle inequality and Definition~\ref{def:DSRI_kernel}.\\
Part ii) For any $t\in\Tbb$, we have
\begin{equation}
	\big(\kernel(t,t)\kernelh(t,t)\big)^{\frac12}
	\le
	\big(\sup_{t\in\Tbb}\kernelh(t,t)\big)^{\frac12}\kernel(t,t)^{\frac12},
\end{equation}
which implies the claim from the Definition~\ref{def:DSRI_kernel}. 
\qed\vspace{-2mm}
\subsection{DSRI Property for High-order Stable Spline Kernels} 
\label{apn:DSRI_SSn}
Let $\beta$ be a positive real number and 
$(x)_+$ denote the non-negative part of $x$, for any $x\in\Rbb$, that is $(x)_+:=\max\{x,0\}$.
With respect to each $n\in\Zbb_+$, the $n^\nth$-order \emph{stable spline} kernel $\kernelSSn{n}:\Rbb_+\times\Rbb_+\to\Rbb$
is defined as 
\begin{equation}
    \kernelSSn{n}(s,t)
    =\int_0^1 
    \frac{(\expe^{-\beta s}-u)_+^{n-1}(\expe^{-\beta t}-u)_+^{n-1}}{\big((n-1)!\big)^2}\drm u,
\end{equation}
for any $s,t\in\Rbb_+$ \cite{pillonetto2014kernel}.
\begin{theorem}\label{thm:kernel_nSS_DSRI}
The $n^\nth$-order stable spline kernel is DSRI.
\end{theorem}
\begin{proof}
For each $t\in\Rbb_+$, one can easily see that
\begin{equation}
\begin{split}
    &\kernelSSn{n}(t,t)
    =
    \int_{0}^{1}
    \frac{(\expe^{-\beta t}-u)_+^{2n-2}}{\big((n-1)!\big)^2}\drm u
    \\
    &\ \ =
    \int_{0}^{\expe^{-\beta t}}
    \!\!\!\!\!\!
    \frac{u^{2n-2}}{\big((n-1)!\big)^2}\drm u
    =
    \frac{\expe^{-(2n-1)\beta t}}{(2n-1)\big((n-1)!\big)^2}.
\end{split}
\end{equation}
Therefore, $\kernelSSn{n}$ is diagonally dominated by kernel 
$\kernelRnE{n}(\cdot,\cdot\,; 1,\expe^{-(2n-1)\beta})$. Thus, due to Theorem~\ref{thm:dominancy-and-k_RnE_DSRI}, 
$\kernelSSn{n}$ is a DSRI kernel.
\end{proof}
\subsection{DSRI Property for Simulation-Induced Kernels} 
\label{apn:DSRI_SI}
Given $\vcv=(v_t)_{t\in\Tbb}$ in $\Lscrone$ with non-negative values,
a stable SISO system
of order $n$ 
with realization $(\mx{A},\vc{b},\vc{c},d)$, and $n$ by $n$ positive-definite matrix $\mx{Q}$, the \emph{simulation-induced} kernel $\kernelSI:\Tbb\times\Tbb\to\Rbb$ is defined such that, for any $s,t\in\Tbb$, we have
\begin{equation}\label{eqn:kSI_DT}
\begin{split}
\kernelSI(s,t)
&=
\vc{c}\mx{A}^s\mx{Q}(\mx{A}^t)^\tr\vc{c}^\tr
+
d^2v_sv_t\onefun_{\{0\}}(s-t)
\\&\!\!\! 
+
v_sd
\sum_{k=0}^{t-1}
\onefun_{\{0\}}(s-k)v_k\vc{b}^\tr(\mx{A}^{t-1-k})^\tr\vc{c}^\tr
\\&\!\!\! 
+
v_td\sum_{k=0}^{s-1}\onefun_{\{0\}}(t-k)v_k
\vc{b}^\tr(\mx{A}^{s-1-k})^\tr\vc{c}^\tr
\\&\!\!\! 
+
\sum_{k=0}^{\min(t,s)-1}
v_k^2\vc{c}\mx{A}^{s-1-k}\vc{b}\vc{b}^\tr(\mx{A}^{t-1-k})^\tr\vc{c}^\tr,
\end{split}    
\end{equation}
when $\Tbb=\Zbb_+$, and 
\begin{equation}\label{eqn:kSI_CT}
\begin{split}
\kernelSI(s,t)
&=
\vc{c}\expe^{\mx{A}s} \mx{Q} (\expe^{\mx{A}t})^\tr \vc{c}^\tr
\\&\!\!\! \!\!\!\!\!\!\!\!
+
\int_0^{\min(s,t)}v_{\tau}^2
\vc{c}\expe^{\mx{A}{(s-\tau)}}\vc{b}
\vc{b}^\tr(\expe^{\mx{A}{(t-\tau)}})^\tr\vc{c}^\tr\, \drm\tau,
\end{split}    
\end{equation}
when $\Tbb=\Rbb_+$ \cite{chen2018kernel}. 
\begin{theorem}\label{thm:kernel_SI_DSRI}
  Let assume that there exist $\gamma_1,\gamma_2\in\Rbb_+$ and $\alpha\in[0,1)$ such that, for any $t\in\Tbb$, 
  we have
  $\|\Phi_t\|\le \gamma_1\, \alpha^t$ and $v_t^2 \le \gamma_2\, \alpha^t$, where $\Phi_t$ denotes matrix $\mx{A}^t$, when $\Tbb=\Zbb_+$, and matrix $\expe^{\mx{A}t}$, when $\Tbb=\Rbb_+$.
  Then, $\kernelSI$ is a DSRI kernel.
\end{theorem}
\begin{proof}
For any $t\in\Tbb$, one can show that
\begin{equation}\label{eqn:3rd-term-DT}
    \Big| \sum_{k=0}^{t-1}
    v_k^2\vc{c}\mx{A}^{t-1-k}\vc{b}\vc{b}^\tr(\mx{A}^{t-1-k})^\tr\vc{c}^\tr  
    \Big|
    \le
    \frac{\gamma \alpha^t}{\alpha(1-\alpha)},
\end{equation}
when $\Tbb=\Zbb_+$, and
\begin{equation}\label{eqn:3rd-term-CT}
    \Big|    
    \int_0^{t}
    v_{\tau}^2\vc{c}\expe^{\mx{A}{(t-\tau)}}\vc{b}
    \vc{b}^\tr(\expe^{\mx{A}{(t-\tau)}})^\tr\vc{c}^\tr\,
    \drm\tau    
    \Big|
    \le -\frac{\gamma \alpha^t}{\ln(\alpha)},
\end{equation}
when $\Tbb=\Rbb_+$, where $\gamma=\|\vc{b}\|^2\|\vc{c}\|^2\gamma_1^2\gamma_2$.
Define the kernel $\kernel:\Tbb\times\Tbb\to\Rbb$ as
\begin{equation*}
    \kernel(s,t)
    =
    \kernelRnE{2}(s,t; [\lambda_1,\lambda_2]^\tr,[\alpha^2, \alpha]^\tr) + d^2\kernel_{\vc{v}}(s,t),
\end{equation*}
for any $s,t\in\Tbb$, where $\lambda_1 := \gamma_1^2\|\mx{Q}\|\|\vc{c}\|^2$,  $\lambda_2:=\frac{\gamma}
{\alpha(1-\alpha)}$,
when $\Tbb=\Zbb_+$,
and 
$\lambda_2:=-\frac{\gamma}{\ln(\alpha)}$,
when $\Tbb = \Rbb_+$. 
According to Theorem~\ref{thm:dominancy-and-k_RnE_DSRI} and Theorem~\ref{thm:sum_product_DSRI}, we know that $\kernel$ is a DSRI kernel.
Moreover, due to \eqref{eqn:3rd-term-DT} and \eqref{eqn:3rd-term-CT}, one can easily see that $\kernelSI$ is diagonally dominated by $\kernel$. Therefore, kernel $\kernelSI$ is DSRI.
\end{proof}
\subsection{DSRI Property and Sampling} 
\label{apn:DSRI_sampling}
We say  $\sigma:\Zbb_+\to\Rbb_+$ is a \emph{proper sampling}  function if 
$\inf_{t\in\Zbb_+}\sigma(t+1)-\sigma(t)>0$. 
The following theorem says that DSRI property is preserved under proper sampling.
\begin{theorem}\label{thm:DSRI_sampling}
Let $\kernel:\Rbb_+\times\Rbb_+\to\Rbb$ be a  positive-definite kernel and $\kernel_{\sigma}:\Zbb_+\times\Zbb_+\to\Rbb$ be defined as 
\begin{equation}
    \kernel_{\sigma}(s,t) = \kernel(\sigma(s),\sigma(t)), \qquad s,t\in\Zbb_+.
\end{equation}
Define function $d_{\kernel}:\Rbb_+\to\Rbb$ as $d_{\kernel}(t)=\kernel(t,t)^{\frac12}$, for any $t$.
If $\kernel$ is a DSRI kernel with non-increasing $d_{\kernel}$, then $\kernel_{\sigma}$ is a DSRI kernel.
\end{theorem}
\begin{proof}
The positive-definiteness of $\kernel_{\sigma}$ is a direct result of the same property for $\kernel$.
Let $\delta$ be defined as 
$\delta := \inf_{t\in\Zbb_+}\sigma(t+1)-\sigma(t)>0$. 
Since $\delta>0$ and $d_{\kernel}$ is a non-increasing function, we have
\begin{equation}
\begin{split}
    \DSRIvalue&(\kernel_{\sigma})
     - \kernel_{\sigma}(0,0)^{\frac12} = 
    \sum_{t=1}^{\infty}\kernel_{\sigma}(t,t)^{\frac12}
    =
    \sum_{t=1}^{\infty}d_{\kernel}(\sigma(t))
    \!\!\!\!\!\!\!\!
    \!\!\!\!\!\!\!\!
    \\&\le
     \frac{1}{\delta}
     \sum_{t=1}^{\infty}d_{\kernel}(\sigma(t))
     \big(\sigma(t)-\sigma(t-1)\big)
    \\&\le
    \frac{1}{\delta}\int_{\Rbb_+}\kernel(s,s)^{\frac12}\drm s
    = \frac{1}{\delta} \DSRIvalue(\kernel)
    <\infty,
\end{split}    
\end{equation}
which implies that  $\kernel_{\sigma}$ is a DSRI kernel.
\end{proof}
\subsection{DSRI Property for Reparameterized Kernels} 
\label{apn:DSRI_reparameterization}
\begin{theorem}\label{thm:reparameterization_DSRI}
    Let $\kernel:\Tbb\times\Tbb\to\Rbb$ be a DSRI kernel
    and 
    $\rho:\Tbb\to\Tbb$ be a strictly increasing function, which is assumed to be differentiable with 
    $\inf_{\tau\in\Rbb_+}\frac{\drm \rho(\tau)}{\drm \tau}>0$, when $\Tbb=\Rbb_+$.
    Define $\kernelh:\Tbb\times\Tbb\to\Rbb$ as $\kernelh(s,t)=\kernel(\rho(s),\rho(t))$, for any $s,t\in\Tbb$. Then, $\kernelh$ is a DSRI kernel.
\end{theorem}
\begin{proof}
The positive-definiteness of $\kernelh$ is directly concluded from the same property of $\kernel$. 
The properties of $\rho$ imply that it has a well-defined inverse function $\rho^{-1}:\rho(\Tbb)\to\Tbb$, which is a strictly increasing map.
Therefore, for any $s\in\rho(\Tbb)$, there exists a unique $t\in\Tbb$ such that $s=\rho(t)$. Accordingly, for the case of $\Tbb=\Zbb_+$, we have
\begin{equation}
\begin{split}	
	\DSRIvalue(\kernelh)
	&=
	\sum_{t\in\Zbb_+}\kernelh(t,t)^{\frac12}
	=
	\sum_{t\in\Zbb_+}\kernel(\rho(t),\rho(t))^{\frac12}
	\\&=
	\sum_{s=\rho(0)}^{\infty}
	\kernel(s,s)^{\frac12}
	\le
	\sum_{s\in\Zbb_+}\kernel(s,s)^{\frac12}
	<\infty,
\end{split}	
\end{equation} 
which implies that $\kernelh$ is DSRI.
Similarly, for the case of $\Tbb=\Rbb_+$, we have
\begin{equation*}
	\begin{split}
		\!\!\!\!\!\!
		\DSRIvalue&(\kernelh)
		=
		\int_{\Rbb_+}\kernelh(t,t)^{\frac12}\drm t
		=
		\int_{\Rbb_+}\kernel(\rho(t),\rho(t))^{\frac12} \drm t
		\!\!\!\!
		\\&\!\!\!\!\!\!=\!\!
		\int_{\rho(0)}^{\infty}\!
		\kernel(s,s)^{\frac12}
		\frac{1}
		{\frac{\drm \rho}{\drm \tau}(\rho^{-1}(s))} \drm s
		\le 
		\frac{\int_{s\in\Rbb_+}\!\!\kernel(s,s)^{\frac12} \drm s}
		{\inf_{\tau\in\Rbb_+}\frac{\drm \rho(\tau)}{\drm \tau}}
		\!\!\!\!
		\\&\!\!\!\!\!\!\le 
		\frac{\DSRIvalue(\kernel)}
		{\inf_{\tau\in\Rbb_+}\frac{\drm \rho(\tau)}{\drm \tau}}
		<\infty.
	\end{split}
\end{equation*}
This concludes the proof. 
\end{proof}
\subsection{Proof of Theorem~\ref{thm:integrable_not_DSRI}} 
\label{apn:proof:thm:integrable_not_DSRI}
Let $\Tbb=\Zbb_+$ and define a symmetric function $\kernel:\Zbb_+\times\Zbb_+\to\Rbb$ such that, for any $s,t\in\Zbb_+$, we have 
\begin{equation}\label{eqn:pf:integrable_not_DSRI_pf:eqn_01}
	\kernel(s,t) = 
	\begin{cases}
		\frac{1}{(1+s)^2},	& \text{ if } s=t,\\
		0,					& \text{ if } s\ne t.\\
	\end{cases}
\end{equation}
For any $t_1,\ldots,t_n\in\Zbb_+$ and any $a_1,\ldots,a_n\in\Rbb$, one can see that
\begin{equation}\label{eqn:pf:integrable_not_DSRI_pf:eqn_02}
	\begin{split}
		\sum_{i=1}^n\sum_{j=1}^n&a_ia_j\kernel(t_i,t_j)=\\
		&\sum_{\substack{0\le t\le\bar{t}\\\Ical_t\ne \emptyset}}
		\frac{1}{(1+t)^2}
		\Big(\sum_{i\in\Ical_t}a_i\Big)^2
		\ge 0,
	\end{split}
\end{equation}
where $\bar{t}=\max\{t_1,\ldots,t_n\}$ and $\Ical_t=\big\{i\in\{1,\ldots,n\}\big|t_i=t\big\}$, for $t=0,\ldots,\bar{t}$.
This implies that $\kernel$ is a  positive-definite kernel. Moreover, we have
\begin{equation}
	\sum_{s\in\Zbb_+}\sum_{t\in\Zbb_+}\kernel(s,t)=\sum_{s\in\Zbb_+}\frac{1}{(1+s)^2}=\frac{1}{6}\pi^2< \infty,
\end{equation}
and
\begin{equation}
	\DSRIvalue(\kernel) = \sum_{s\in\Zbb_+}\kernel(s,s)^{\frac12}=\sum_{s\in\Zbb_+}\frac{1}{1+s}=\infty.
\end{equation}
Therefore, $\kernel$ is an integrable  positive-definite kernel which is not DSRI. Let $\Tbb=\Rbb_+$ and function $f:\Rbb_+\to\Rbb_+$ be defined as
\begin{equation}
	f(t) = 
	\begin{cases}
		\cos(\pi t),	& \text{ if } t\in[0,\frac12],\\
		0,				& \text{ if } t\ge \frac12,\\
	\end{cases}
\end{equation}
for any $t\in\Rbb_+$. Note that $f$ is a continuous and positive function. Define  $\kernelh:\Zbb_+\times\Zbb_+\to\Rbb$ such that, for any $s,t\in\Zbb_+$, we have $\kernelh(s,t)=\kernel(\floor{s},\floor{t})g(s)g(t)$,
where $\kernel$ is introduced in \eqref{eqn:pf:integrable_not_DSRI_pf:eqn_01},
and function
$g:\Rbb_+\to\Rbb_+$ is defined as $g(s)=f(s-\floor{s})$, for any $s\in\Rbb_+$.
One can easily see that $\kernelh$ is continuous. 
Moreover, for any $t_1,\ldots,t_n\in\Zbb_+$ and any  $b_1,\ldots,b_n\in\Rbb$, we have
\begin{equation*} 
	\begin{split}
		\sum_{i,j=1}^n b_ib_j\kernelh(t_i,t_j)
		&=
		\sum_{i,j=1}^n b_ib_j\kernel(\floor{t_i},\floor{t_j})g(t_i)g(t_j)
		\\&=
		\sum_{i,j=1}^n a_ia_j\kernel(\floor{t_i},\floor{t_j}),
	\end{split}
\end{equation*}
where $a_i$ is defined as $a_i:=b_i g(t_i)$, for $i=1,\ldots,n$. 
Therefore, due to \eqref{eqn:pf:integrable_not_DSRI_pf:eqn_02}, we have 
$\sum_{i,j=1}^n b_ib_j\kernelh(t_i,t_j)\ge 0$,
which implies that $\kernelh$ is a positive-definite kernel. We know that
\begin{equation}\!\!\!\!\!
	\begin{split}
		\ \ &\!\!\!\!\!
		\int_{\Rbb_+}\!\int_{\Rbb_+}\!|\kernelh(s,t)|\drm s\drm t
		=
		\int_{\Rbb_+}\!\int_{\Rbb_+}\!\kernel(\floor{s},\floor{t})g(s)g(t)\drm s\drm t \!\!\!\!\!\!\!\!\!\!\!\!\!\!\!\!\!
		\\&\le
		\int_{\Rbb_+}\int_{\Rbb_+}\kernel(\floor{s},\floor{t})\drm s\drm t
		\\&=
		\sum_{n\in\Zbb_+}
		\int_{n}^{n+1}\int_{n}^{n+1}\kernel(\floor{s},\floor{t})\drm s\drm t
		\\&=
		\sum_{n\in\Zbb_+}
		\kernel(n,n)
		=
		\sum_{n\in\Zbb_+}
		\frac{1}{(n+1)^2} =\frac16\pi^2,
	\end{split}	
\end{equation} 
which implies that $\kernelh$ is integrable.
On the other hand, we have
\begin{equation}
	\begin{split}
		\DSRIvalue(\kernelh)
		&=
		\int_{\Rbb_+}\kernel(\floor{s},\floor{s})^\frac12 g(s)\drm s
		\\&
		=
		\sum_{n\in\Zbb_+}
		\int_{n}^{n+1}\kernel(\floor{s},\floor{s})^\frac12 g(s)\drm s,	
		\end{split}	
\end{equation}
and thus, from definition of $\kernel$, it follows that
\begin{equation}
	\begin{split}
		\DSRIvalue(\kernelh)&=
		\sum_{n\in\Zbb_+}
		\kernel(n,n)^\frac12\int_{0}^{1}f(s)\drm s
		\\&=
		\Big(\sum_{n\in\Zbb_+}\frac{1}{n+1}\Big)\int_{0}^{1}f(s)\drm s = \infty.
	\end{split}	
\end{equation} 
Therefore, $\kernelh$ is not a DSRI kernel. 
\qed\vspace{-2mm}
\subsection{Proof of Theorem~\ref{thm:continuity_inherits}} 
\label{apn:proof:thm:continuity_inherits}
	The first part of the theorem is due to Theorem~\ref{thm:DSRI_stable}.
	For the second part of the theorem, we only provide the proof for the case of $\Tbb=\Rbb_+$. The proof for $\Tbb=\Zbb_+$ is similar. 
	
	Let $\vcg=(g_s)_{s\in\Rbb_+}$. Due to the reproducing property,  we have $g_s=\inner{\vcg}{\kernel_s}_{\Hk}$ and $\|\kernel_s\|_{\Hk}^2=\inner{\kernel_s}{\kernel_s}_{\Hk}=\kernel(s,s)$, for any $s\in\Rbb_+$.
	Subsequently, from the Cauchy-Schwartz inequality, it follows that
	\begin{equation*}
		|g_s| = |\inner{\vcg}{\kernel_s}_{\Hk}|
		\le 
		\|\vcg\|_{\Hk} \!\ \|\kernel_s\|_{\Hk} 
		= 
		\|\vcg\|_{\Hk} \!\ \kernel(s,s)^{\frac12}.
	\end{equation*} 
	Accordingly, since $\|\vcg\|_1=\int_{\Rbb_+}|g_s|\drm s$, we have
	\begin{equation}\label{eqn:pf_continuity_inherits_thm:eqn_1}
	\begin{split}
		\|\vcg\|_1
		\le \int_{\Rbb_+}\|\vcg\|_{\Hk} \!\ \kernel(s,s)^{\frac12}\drm s 
		=\DSRIvalue(\kernel)\|\vcg\|_{\Hk}.
	\end{split}	
	\end{equation}
	On the other hand, from the definition of operator norm, it follows that
	\begin{equation}\label{eqn:pf_continuity_inherits_thm:eqn_2}
		\|\mxL(\vcg)\|_{\Banach}
		\le
		\|\mxL\|_{\Lcal(\Lscrone,\Banach)}
		\|\vcg\|_1.
	\end{equation}
	Considering \eqref{eqn:pf_continuity_inherits_thm:eqn_1},  we know that
	\begin{equation}\label{eqn:pf_continuity_inherits_thm:eqn_3}
		\|\mxL(\vcg)\|_{\Banach}
		\le 
		\|\mxL\|_{\Lcal(\Lscrone,\Banach)} \DSRIvalue(\kernel)\|\vcg\|_{\Hk}.
	\end{equation}
	Therefore, due to \eqref{eqn:pf_continuity_inherits_thm:eqn_3} and the definition of operator norm, we have
	\begin{equation}\label{eqn:pf_continuity_inherits_thm:eqn_4}
		\begin{split}
			\|\mxL&\|_{\Lcal(\Hk,\Banach)}
			=
			\supOp_{\substack{\vcg\in\Hk \\ \|\vcg\|_{\Hk}\le 1}}
			\|\mxL(\vcg)\|_{\Banach}
			\\&\le 
			\supOp_{\substack{\vcg\in\Hk \\ \|\vcg\|_{\Hk}\le 1}}
			\Big( \|\mxL\|_{\Lcal(\Lscrone,\Banach)} \DSRIvalue(\kernel)\|\vcg\|_{\Hk}\Big)
			\\&=	
			\|\mxL\|_{\Lcal(\Lscrone,\Banach)} \DSRIvalue(\kernel),
		\end{split}	
	\end{equation}
	which implies \eqref{eqn:continuity_inherits_thm:norm_ineq} and concludes the proof.
\qed\vspace{-2mm}
\subsection{Proof of Theorem~\ref{thm:Lv_continuity}} 
\label{apn:proof:thm:Lv_continuity}
By an abuse of notation, we define $\mxL:\Lscrone\to\Banach$ similarly to \eqref{eqn:Lv_operator}. 
According to \cite[Theorem 8.2]{mikusinski1978Bochner}, $(g_tv_t)_{t\in\Tbb}$ is a Bochner integrable function, for any $\vcg=(g_t)_{t\in\Tbb}\in\Lscrone$.  This implies that
$\mxL:\Lscrone\to\Banach$ is a well-defined linear operator. Furthermore, we have 
\begin{equation}
    \|\mxL(\vcg)\|_{\Banach} \le 
    \|\vcg\|_1 \big(
    \esssup_{t\in\Tbb}\|v_t\|_{\Banach}\big),
\end{equation}
for any $\vcg=(g_t)_{t\in\Tbb}\in\Lscrone$.
Therefore, one can see
\begin{equation*}
		\|\mxL\|_{\Lcal(\Lscrone,\Banach)}
		=
		\supOp_{\substack{\vcg\in\Lscrone \\ \|\vcg\|_{\Lscrone}\le 1}}
		 \|\mxL(\vcg)\|_{\Banach}
		\le 
		\esssup_{t\in\Tbb}\|v_t\|_{\Banach}, 
\end{equation*}
i.e., $\mxL:\Lscrone\to\Banach$ is a continuous linear operator. Thus, the claim follows directly from Theorem~\ref{thm:continuity_inherits}.
\qed\vspace{-2mm}
\subsection{Proof of Lemma~\ref{lem:GP_stable}} 
\label{apn:proof:lem:GP_stable}

We prove the lemma for the case of $\Tbb=\Rbb_+$. The proof for $\Tbb=\Zbb_+$ follows the same line of argument.
Note that we have
\begin{equation*}
    \big\{\omega\in\Omega 
    \,\big|\, 
    \|\vcg(\omega)\|_1<\infty
    \big\}
    =
    \CupOp_{r=1}^{\infty}
    \big\{\omega\in\Omega
    \,\big|\,
    \|\vcg(\omega)\|_1\le r
    \big\}.
\end{equation*}
Accordingly, from the sub-additivity property of $\Pbb$, we know that 
\begin{equation*} 
    0<\Pbb\big[\|\vcg\|_1<\infty\big]\le \sum_{r=1}^{\infty}\,\Pbb\big[\|\vcg\|_1\le r\big].
\end{equation*}
Therefore, there exists $r\in\Nbb$ such that, for event $A$ defined as $A:=\big\{\omega\in\Omega\,\big|\,\|\vcg(\omega)\|_1\le r\big\}$, we have $\gamma:=\Pbb[A]>0$.
Accordingly, due to the properties of indicator functions, the definition of $A$, and the Tonelli's Theorem \cite{stein2009real},  we can see that
\begin{equation}\label{eqn:pf_eq3}
\begin{split}
    \gamma
    =\Pbb[A]  
    &=
    \frac1r\Ebb\big[r\onefun_{A}\big]
    \\&\ge 
    \frac1r
    \Ebb\big[\|\vcg\|_1\onefun_{A}\big]
    =
    \int_{\Rbb_+}
    \Ebb\big[|g_t|\onefun_{A}\big] \drm t.
\end{split}
\end{equation}
With respect to each $t\in\Rbb_+$, define event $B_t$ as
\begin{equation}\label{eqn:pf_eq4}
    B_t
    = 
    \big\{\omega\in\Omega 
    \,\big|\,   
    |g_t(\omega)|\le\epsilon\kernel(t,t)^{\frac12} \big\},  
\end{equation}
where $\epsilon$ is the positive real number characterized as $\epsilon:=\Phi^{-1}(\frac12 \gamma)$.
For each $t\in\Rbb_+$, we have $A\supseteq A\cap B_t^{\mathrm{c}}$.  
Therefore, from \eqref{eqn:pf_eq3} and \eqref{eqn:pf_eq4}, it follows that
\begin{equation}\label{eqn:pf_eq5}
\begin{split}
    \gamma\ 
    &\ge    
    \int_{\Rbb_+}
    \Ebb\big[|g_t|\onefun_{A}\big] \drm t
    \\&\ge
    \int_{\Rbb_+}
    \Ebb\big[|g_t|\onefun_{A\cap B_t^{\mathrm{c}}}\big] \drm t
    \\&\ge
    \int_{\Rbb_+}    
    \epsilon\kernel(t,t)^{\frac12}
    \Ebb\big[\onefun_{A\cap B_t^{\mathrm{c}}}\big] \drm t.
\end{split}
\end{equation}
Moreover, for each $t\in\Rbb_+$, we have 
$\onefun_{A\cap B_t^{\mathrm{c}}} \ge \onefun_{A}-\onefun_{B_t}$, which implies that 
$\Ebb\big[\onefun_{A\cap B_t^{\mathrm{c}}}\big]
\ge \Pbb[A]-\Pbb[B_t]$.
Subsequently, from \eqref{eqn:pf_eq3} and \eqref{eqn:pf_eq5}, we can see that
\begin{equation}\label{eqn:pf_eq6}
    \gamma
    \ge
    \int_{\Rbb_+} \!\!    
    \epsilon\kernel(t,t)^{\frac12}
    \big(\Pbb[A]-\Pbb[B_t]\big) \drm t.
\end{equation}
We know that $g_t\sim\Ncal(0,\kernel(t,t)^{\frac12})$,  $t\in\Rbb_+$. Accordingly, from the definition of sets $A$ and $B_t$, we have
\begin{equation*}\label{eqn:pf_eq7}
\begin{split}
    \epsilon\kernel(t,t)^{\frac12}\big(\Pbb[A]-\Pbb[B_t]\big) 
    &=
    \epsilon\kernel(t,t)^{\frac12}\big(\Pbb[A]-\Phi(\epsilon)\big)     
    \\&=
    \frac12\epsilon   \kernel(t,t)^{\frac12}\gamma.    
\end{split}
\end{equation*}
Therefore, \eqref{eqn:pf_eq6} implies that
\begin{equation}\label{eqn:pf_eq8}
    \gamma
    \ge
    \frac12\epsilon \gamma
    \int_{\Rbb_+} \!\!    
       \kernel(t,t)^{\frac12} \drm t 
    = 
    \frac12\epsilon \gamma
    \DSRIvalue(\kernel),
\end{equation}
and subsequently, we have $\DSRIvalue(\kernel)<\infty$, and $\kernel$ is a DSRI kernel.
Furthermore, from Lemma~\ref{lem:DSRI_GP_stable}, it follows that 
$\Pbb[\|\vcg\|_1<\infty]=1$, which concludes the proof.
\qed\vspace{-2mm}
\subsection{Proof of Theorem~\ref{thm:GP_stable}} 
\label{apn:proof:thm:GP_stable}
Note that $\vcg\sim\GP(\vcm,\kernel)$ if and only if
$\vcg-\vcm\sim\GP(\zero,\kernel)$.
Since $\vcm$ is a stable impulse response,  the stability of $\vcg$ is equivalent to the stability of $\vcg-\vcm$. 
Accordingly, the claim follows from Lemma~\ref{lem:DSRI_GP_stable} and Lemma~\ref{lem:GP_stable}.
\qed\vspace{-2mm}

\bibliographystyle{IEEEtran}        

{
\scriptsize{ 
	\bibliography{mainbib_DSRI}    

\begin{thebibliography}{10}
\providecommand{\url}[1]{#1}
\csname url@samestyle\endcsname
\providecommand{\newblock}{\relax}
\providecommand{\bibinfo}[2]{#2}
\providecommand{\BIBentrySTDinterwordspacing}{\spaceskip=0pt\relax}
\providecommand{\BIBentryALTinterwordstretchfactor}{4}
\providecommand{\BIBentryALTinterwordspacing}{\spaceskip=\fontdimen2\font plus
\BIBentryALTinterwordstretchfactor\fontdimen3\font minus
  \fontdimen4\font\relax}
\providecommand{\BIBforeignlanguage}[2]{{%
\expandafter\ifx\csname l@#1\endcsname\relax
\typeout{** WARNING: IEEEtran.bst: No hyphenation pattern has been}%
\typeout{** loaded for the language `#1'. Using the pattern for}%
\typeout{** the default language instead.}%
\else
\language=\csname l@#1\endcsname
\fi
#2}}
\providecommand{\BIBdecl}{\relax}
\BIBdecl

\bibitem{aronszajn1950theory}
N.~Aronszajn, ``Theory of reproducing kernels,'' \emph{Transactions of the
  American Mathematical Society}, vol.~68, no.~3, pp. 337--404, 1950.

\bibitem{parzen1959statistical}
E.~Parzen, ``Statistical inference on time series by {Hilbert} space methods,
  i,'' Department of Statistics, Stanford University, Technical Report No. 23,
  Tech. Rep., 1959.

\bibitem{wahba1990spline}
G.~Wahba, \emph{{Spline Models for Observational Data}}.\hskip 1em plus 0.5em
  minus 0.4em\relax SIAM, 1990.

\bibitem{cucker2002best}
F.~Cucker and S.~Smale, ``Best choices for regularization parameters in
  learning theory: {O}n the bias-variance problem,'' \emph{Foundations of
  Computational Mathematics}, vol.~2, no.~4, pp. 413--428, 2002.

\bibitem{berlinet2011reproducing}
A.~Berlinet and C.~Thomas-Agnan, \emph{Reproducing {K}ernel {H}ilbert {S}paces
  in {P}robability and {S}tatistics}.\hskip 1em plus 0.5em minus 0.4em\relax
  Springer Science and Business Media, 2011.

\bibitem{khosravi2021Koopman}
M.~Khosravi, ``Representer theorem for learning {K}oopman operators,''
  \emph{IEEE Transactions on Automatic Control}, 2023.

\bibitem{kimeldorf1970correspondence}
G.~S. Kimeldorf and G.~Wahba, ``A correspondence between {B}ayesian estimation
  on stochastic processes and smoothing by splines,'' \emph{The Annals of
  Mathematical Statistics}, vol.~41, no.~2, pp. 495--502, 1970.

\bibitem{lukic2001stochastic}
M.~Luki{\'c} and J.~Beder, ``Stochastic processes with sample paths in
  reproducing kernel {H}ilbert spaces,'' \emph{Transactions of the American
  Mathematical Society}, vol. 353, no.~10, pp. 3945--3969, 2001.

\bibitem{kanagawa2018gaussian}
M.~Kanagawa, P.~Hennig, D.~Sejdinovic, and B.~K. Sriperumbudur, ``Gaussian
  processes and kernel methods: {A} review on connections and equivalences,''
  \emph{arXiv preprint arXiv:1807.02582}, 2018.

\bibitem{cuckerANDsmale2002mathematical}
F.~Cucker and S.~Smale, ``On the mathematical foundations of learning,''
  \emph{American Mathematical Society}, vol.~39, no.~1, pp. 1--49, 2002.

\bibitem{zadeh1956identification}
L.~Zadeh, ``On the identification problem,'' \emph{IRE Transactions on Circuit
  Theory}, vol.~3, no.~4, pp. 277--281, 1956.

\bibitem{ljung2010perspectives}
L.~Ljung, ``Perspectives on system identification,'' \emph{Annual Reviews in
  Control}, vol.~34, no.~1, pp. 1--12, 2010.

\bibitem{schoukens2019nonlinear}
J.~Schoukens and L.~Ljung, ``Nonlinear system identification: A user-oriented
  road map,'' \emph{IEEE Control Systems Magazine}, vol.~39, no.~6, pp. 28--99,
  2019.

\bibitem{khosravi2021ROA}
M.~Khosravi and R.~S. Smith, ``Nonlinear system identification with prior
  knowledge on the region of attraction,'' \emph{IEEE Control Systems Letters},
  vol.~5, no.~3, pp. 1091--1096, 2021.

\bibitem{ahmadi2020learning}
A.~A. Ahmadi and B.~El~Khadir, ``Learning dynamical systems with side
  information (short version),'' \emph{Proceedings of Machine Learning
  Research}, vol. 120, pp. 718--727, 2020.

\bibitem{khosravi2021grad}
M.~Khosravi and R.~S. Smith, ``Convex nonparametric formulation for
  identification of gradient flows,'' \emph{IEEE Control Systems Letters},
  vol.~5, no.~3, pp. 1097--1102, 2021.

\bibitem{pillonetto2010new}
G.~Pillonetto and G.~De~Nicolao, ``A new kernel-based approach for linear
  system identification,'' \emph{Automatica}, vol.~46, no.~1, pp. 81--93, 2010.

\bibitem{ljung2020shift}
L.~Ljung, T.~Chen, and B.~Mu, ``A shift in paradigm for system
  identification,'' \emph{International Journal of Control}, vol.~93, no.~2,
  pp. 173--180, 2020.

\bibitem{pillonetto2014kernel}
G.~Pillonetto, F.~Dinuzzo, T.~Chen, G.~De~Nicolao, and L.~Ljung, ``Kernel
  methods in system identification, machine learning and function estimation: A
  survey,'' \emph{Automatica}, vol.~50, no.~3, pp. 657--682, 2014.

\bibitem{chiuso2019system}
A.~Chiuso and G.~Pillonetto, ``System identification: A machine learning
  perspective,'' \emph{Annual Review of Control, Robotics, and Autonomous
  Systems}, vol.~2, pp. 281--304, 2019.

\bibitem{khosravi2021robust}
M.~Khosravi and R.~S. Smith, ``On robustness of kernel-based regularized system
  identification,'' \emph{IFAC-PapersOnLine}, vol.~54, no.~7, pp. 749--754,
  2021, {I}FAC Symposium on System Identification.

\bibitem{fujimoto2017extension}
Y.~Fujimoto, I.~Maruta, and T.~Sugie, ``Extension of first-order stable spline
  kernel to encode relative degree,'' \emph{IFAC-PapersOnLine}, vol.~50, no.~1,
  pp. 14\,016--14\,021, 2017.

\bibitem{zheng2021bayesian}
M.~Zheng and Y.~Ohta, ``{Bayesian positive system identification: {T}runcated
  {G}aussian prior and hyperparameter estimation},'' \emph{Systems and Control
  Letters}, vol. 148, p. 104857, 2021.

\bibitem{darwish2018quest}
M.~A.~H. Darwish, G.~Pillonetto, and R.~T{\'o}th, ``The quest for the right
  kernel in {B}ayesian impulse response identification: {T}he use of {OBF}s,''
  \emph{Automatica}, vol.~87, pp. 318--329, 2018.

\bibitem{khosravi2019positive}
M.~Khosravi and R.~S. Smith, ``Kernel-based identification of positive
  systems,'' in \emph{Conference on Decision and Control}, 2019, pp.
  1740--1745.

\bibitem{prando2017maximum}
G.~Prando, A.~Chiuso, and G.~Pillonetto, ``Maximum entropy vector kernels for
  {MIMO} system identification,'' \emph{Automatica}, vol.~79, pp. 326--339,
  2017.

\bibitem{fujimoto2018kernel}
Y.~Fujimoto and T.~Sugie, ``{Kernel-based impulse response estimation with a
  priori knowledge on the DC gain},'' \emph{IEEE control systems letters},
  vol.~2, no.~4, pp. 713--718, 2018.

\bibitem{risuleo2019bayesian}
R.~S. Risuleo, F.~Lindsten, and H.~Hjalmarsson, ``Bayesian nonparametric
  identification of {W}iener systems,'' \emph{Automatica}, vol. 108, p. 108480,
  2019.

\bibitem{everitt2018empirical}
N.~Everitt, G.~Bottegal, and H.~Hjalmarsson, ``An empirical {B}ayes approach to
  identification of modules in dynamic networks,'' \emph{Automatica}, vol.~91,
  pp. 144--151, 2018.

\bibitem{risuleo2017nonparametric}
R.~S. Risuleo, G.~Bottegal, and H.~Hjalmarsson, ``A nonparametric kernel-based
  approach to {H}ammerstein system identification,'' \emph{Automatica},
  vol.~85, pp. 234--247, 2017.

\bibitem{khosravi2021POS}
M.~Khosravi and R.~S. Smith, ``Regularized identification with internal
  positivity side-information,'' \emph{arXiv preprint arXiv:2111.00407}, 2021.

\bibitem{khosravi2021SSG}
------, ``Kernel-based impulse response identification with side-information on
  steady-state gain,'' \emph{IEEE Transactions on Automatic Control}, 2023.

\bibitem{khosravi2021FDI}
------, ``Kernel-based identification with frequency domain side-information,''
  \emph{Automatica}, vol. 150, p. 110813, 2023.

\bibitem{lataire2016transfer}
J.~Lataire and T.~Chen, ``Transfer function and transient estimation by
  {G}aussian process regression in the frequency domain,'' \emph{Automatica},
  vol.~72, pp. 217--229, 2016.

\bibitem{scandella2021kernel}
M.~Scandella, M.~Mazzoleni, S.~Formentin, and F.~Previdi, ``Kernel-based
  identification of asymptotically stable continuous-time linear dynamical
  systems,'' \emph{International Journal of Control}, pp. 1--14, 2021.

\bibitem{pillonetto2019stable}
G.~Pillonetto, A.~Chiuso, and G.~De~Nicolao, ``Stable spline identification of
  linear systems under missing data,'' \emph{Automatica}, vol. 108, p. 108493,
  2019.

\bibitem{bisiacco2020mathematical}
M.~Bisiacco and G.~Pillonetto, ``{On the mathematical foundations of stable
  RKHSs},'' \emph{Automatica}, vol. 118, p. 109038, 2020.

\bibitem{pillonetto2021sample}
G.~Pillonetto and A.~Scampicchio, ``Sample complexity and minimax properties of
  exponentially stable regularized estimators,'' \emph{IEEE Transactions on
  Automatic Control}, 2021.

\bibitem{bisiacco2020kernel}
M.~Bisiacco and G.~Pillonetto, ``Kernel absolute summability is sufficient but
  not necessary for {RKHS} stability,'' \emph{SIAM Journal on Control and
  Optimization}, vol.~58, no.~4, pp. 2006--2022, 2020.

\bibitem{khosravi2022Lut}
M.~Khosravi and R.~S. Smith, ``The existence and uniqueness of solutions for
  kernel-based system identification,'' \emph{Automatica}, vol. 148, p. 110728,
  2023.

\bibitem{dinuzzo2015kernels}
F.~Dinuzzo, ``Kernels for linear time invariant system identification,''
  \emph{SIAM Journal on Control and Optimization}, vol.~53, no.~5, pp.
  3299--3317, 2015.

\bibitem{zorzi2021second}
M.~Zorzi, ``A second-order generalization of {TC} and {DC} kernels,''
  \emph{arXiv preprint arXiv:2109.09562}, 2021.

\bibitem{chen2018continuous}
T.~Chen, ``Continuous-time {DC} kernel -- a stable generalized first-order
  spline kernel,'' \emph{IEEE Transactions on Automatic Control}, vol.~63,
  no.~12, pp. 4442--4447, 2018.

\bibitem{andersen2020smoothing}
M.~S. Andersen and T.~Chen, ``Smoothing splines and rank structured matrices:
  {R}evisiting the spline kernel,'' \emph{SIAM Journal on Matrix Analysis and
  Applications}, vol.~41, no.~2, pp. 389--412, 2020.

\bibitem{chen2014system}
T.~Chen, M.~S. Andersen, L.~Ljung, A.~Chiuso, and G.~Pillonetto, ``System
  identification via sparse multiple kernel-based regularization using
  sequential convex optimization techniques,'' \emph{IEEE Transactions on
  Automatic Control}, vol.~59, no.~11, pp. 2933--2945, 2014.

\bibitem{hong2018multiple-SURE}
S.~Hong, B.~Mu, F.~Yin, M.~S. Andersen, and T.~Chen, ``Multiple kernel based
  regularized system identification with {SURE} hyper-parameter estimator,''
  \emph{IFAC-papersonline}, vol.~51, no.~15, pp. 13--18, 2018.

\bibitem{khosravi2020low}
M.~Khosravi, M.~Yin, A.~Iannelli, A.~Parsi, and R.~S. Smith, ``Low-complexity
  identification by sparse hyperparameter estimation,''
  \emph{IFAC-PapersOnLine}, vol.~53, no.~2, pp. 412--417, 2020, {I}FAC World
  Congress 2020.

\bibitem{khosravi2020regularized}
M.~Khosravi, A.~Iannelli, M.~Yin, A.~Parsi, and R.~S. Smith, ``Regularized
  system identification: A hierarchical {B}ayesian approach,''
  \emph{IFAC-PapersOnLine}, vol.~53, no.~2, pp. 406--411, 2020, {I}FAC World
  Congress 2020.

\bibitem{chen2018kernel}
T.~Chen, ``On kernel design for regularized {LTI} system identification,''
  \emph{Automatica}, vol.~90, pp. 109--122, 2018.

\bibitem{zorzi2018harmonic}
M.~Zorzi and A.~Chiuso, ``The harmonic analysis of kernel functions,''
  \emph{Automatica}, vol.~94, pp. 125--137, 2018.

\bibitem{marconato2016filter}
A.~Marconato, M.~Schoukens, and J.~Schoukens, ``Filter-based regularisation for
  impulse response modelling,'' \emph{IET Control Theory and Applications},
  vol.~11, no.~2, pp. 194--204, 2016.

\bibitem{brezis2011functional}
H.~Br{\'e}zis, ``{Functional Analysis, Sobolev Spaces and Partial Differential
  Equations},'' 2011.

\bibitem{chen2012estimation}
T.~Chen, H.~Ohlsson, and L.~Ljung, ``On the estimation of transfer functions,
  regularizations and {G}aussian processes -- {R}evisited,'' \emph{Automatica},
  vol.~48, no.~8, pp. 1525--1535, 2012.

\bibitem{pillonetto2016AtomicNuclearKernel}
G.~Pillonetto, T.~Chen, A.~Chiuso, G.~D. Nicolao, and L.~Ljung, ``Regularized
  linear system identification using atomic, nuclear and kernel-based norms:
  The role of the stability constraint,'' \emph{Automatica}, vol.~69, pp.
  137--149, 2016.

\bibitem{chen2018stability}
T.~Chen and G.~Pillonetto, ``On the stability of reproducing kernel {H}ilbert
  spaces of discrete-time impulse responses,'' \emph{Automatica}, vol.~95, pp.
  529--533, 2018.

\bibitem{carmeli2006vector}
C.~Carmeli, E.~De~Vito, and A.~Toigo, ``Vector valued reproducing kernel
  {H}ilbert spaces of integrable functions and {M}ercer theorem,''
  \emph{Analysis and Applications}, vol.~4, no.~4, pp. 377--408, 2006.

\bibitem{mikusinski1978Bochner}
J.~Mikusi{\'n}ski, \emph{The {B}ochner {I}ntegral}.\hskip 1em plus 0.5em minus
  0.4em\relax Birkh\"auser Verlag, Basel--Stuttgart, 1978.

\bibitem{stein2009real}
E.~M. Stein and R.~Shakarchi, \emph{Real analysis: measure theory, integration,
  and Hilbert spaces}.\hskip 1em plus 0.5em minus 0.4em\relax Princeton
  University Press, 2009.

\end{thebibliography}
}}								
\end{document}